\documentclass[10pt,journal,compsoc]{IEEEtran}
\usepackage{amsmath,amsthm,amssymb,paralist,subfigure,cite,graphicx,color,amsbsy,float}
\usepackage{stmaryrd,mathrsfs,url}
\usepackage[table]{xcolor}
\usepackage{cuted,mathtools,lipsum}
\usepackage[ruled,vlined,linesnumbered]{algorithm2e}
\usepackage{soul}
\newtheorem{lem}{Lemma}

\newtheorem{rem}{Remark}
\newtheorem{prop}{Proposition}

\def\mb{\mathbf}

\def\mc{\mathcal}

\newcommand\teL{\mathrel{=\!\!\mathop:}}

\begin{document}
\title{Distributed Detection and Mitigation of Biasing Attacks over Multi-Agent Networks}

\author{Mohammadreza~Doostmohammadian,~\IEEEmembership{Member,~IEEE,} Houman~Zarrabi,~\IEEEmembership{Member,~IEEE,} Hamid~R.~Rabiee,~\IEEEmembership{Senior~Member,~IEEE,}  Usman~A.~Khan,~\IEEEmembership{Senior~Member,~IEEE,} and Themistoklis~Charalambous,~\IEEEmembership{Senior~Member,~IEEE}
	
\IEEEcompsocitemizethanks{\IEEEcompsocthanksitem M.~Doostmohammadian is with
	Faculty of Mechanical Engineering, Semnan University, Semnan, Iran, and School of Electrical Engineering, Aalto University, Espoo, Finland. E-mail: mohammadreza.doostmohammadian@aalto.fi,~doost@semnan.ac.ir
	\IEEEcompsocthanksitem H. Zarrabi is with Iran Telecommunication Research Center (ITRC), Tehran,
	Iran.	E-mail: h.zarrabi@itrc.ac.ir
	\IEEEcompsocthanksitem H. R. Rabiee is with School of Computer Engineering, Sharif University of Technology, Tehran, Iran.	E-mail: rabiee@sharif.edu
	\IEEEcompsocthanksitem U. A. Khan is with Electrical and Computer Engineering Department, Tufts University, Medford, MA, USA.	E-mail: khan@ece.tufts.edu
	\IEEEcompsocthanksitem T. Charalambous is with the School of Electrical Engineering at Aalto University,
	Espoo, Finland. E-mail: themistoklis.charalambous@aalto.fi.
	}
	\thanks{The work of UAK was supported in part by NSF under awards~CMMI-1903972 and~CBET-1935555.}
}

\markboth{IEEE Transaction on Network Science and Engineering}%
{Doostmohammadian \MakeLowercase{\textit{et al.}}: Distributed Detection and Mitigation of Biasing Attacks over Multi-Agent Networks}

\IEEEtitleabstractindextext{
\begin{abstract}
This paper proposes a distributed attack detection and mitigation technique based on distributed estimation over a multi-agent network, where the agents take partial system measurements susceptible to (possible) biasing attacks. In particular, we assume that the system is not locally observable via the measurements in the direct neighborhood of any agent. First, for performance analysis in the attack-free case, 
we show that the proposed distributed estimation is unbiased with bounded mean-square deviation in steady-state. Then, we propose a residual-based strategy to locally detect possible attacks at agents. In contrast to the deterministic thresholds in the literature assuming an upper bound on the noise support, we define the thresholds on the residuals in a probabilistic sense.  
After detecting and isolating the attacked agent, a system-digraph-based mitigation strategy is proposed to replace the attacked measurement with a new observationally-equivalent one to recover potential observability loss. We adopt a graph-theoretic method to classify the agents based on their measurements, to distinguish between the agents recovering the system rank-deficiency and the ones recovering output-connectivity of the system digraph. The attack detection/mitigation strategy is specifically described for each type, which is  of polynomial-order complexity for large-scale applications. Illustrative simulations support our theoretical results.
\end{abstract}
\begin{IEEEkeywords}
Biasing Attacks, False-Data Injection, Distributed Observability, Distributed Estimation, Structural Analysis
\end{IEEEkeywords}}
\maketitle
\IEEEdisplaynontitleabstractindextext
\IEEEpeerreviewmaketitle

\IEEEraisesectionheading{\section{Introduction}\label{sec_intro}}
\IEEEPARstart{D}{ata} 
(or measurements) regarding many real-world  systems, such as wireless sensor networks, multi-agent robotic systems, block-chain and cloud-computing, smart energy networks, are naturally distributed over large geographical regions~\cite{asefi2021application,camsap11}. Collecting all these to a central coordinator (or a fusion center) for the purposes of processing and learning is tedious and impractical in many applications. Distributed learning or inference is thus typically preferred, due to the fact that it does not require long-range communication to a central unit. The corresponding distributed strategies are practically feasible as they rely on local data processing and local communication only among the neighboring agents. However, such decentralized strategies are vulnerable to malicious attacks. In this paper, we consider distributed  detection and mitigation of biasing attacks at sensors/agents
performing distributed estimation over a large-scale dynamical system. Potential applications include secure distributed estimation over Cyber-Physical-Systems (CPS) \cite{yang2020distributed,acc13,isj2020,xu2016distributed,pasqualetti2013attack,shiri2018distributed,asilomar14}, Internet-of-Things (IoT) \cite{chen2018internet,csl2020,pandey2020fault}, smart cities \cite{guo2020unsupervised}, social networks \cite{pequito_gsip,SNAM20,icas21_attack}, and power-grid monitoring systems \cite{dehghani2020deep,cui2012coordinated,rawat2015detection,camsap11,drayer2019detection,chakravorti2017detection,luo2019detection,babu2016optimal,liu2011false,6712169} among others. 

\textcolor{blue}{In distributed estimation (or filtering) applications \cite{usman_tsp:07,cattivelli2008diffusion,flock} a \textit{multi-agent network} is referred to a group of agents with sensing, data-processing, and communication capabilities, which take (noisy) output or measurements of the dynamical system, share their information over a network, and process the received data locally to track the system state. In case of  erroneous or biased data \cite{deghat2019detection,milovsevivc2017analysis}, the distributed estimation performance is significantly degraded if the biased measurements are necessary for \textit{observability}. Recall that observability refers to the possibility of inferring the (entire) states of the dynamical system via tracking outputs/measurements of a \textit{subset} of states over a finite time. This is more challenging in \textit{single time-scale} estimation with only one step of data-fusion between every two consecutive time-steps of system dynamics, and with no local observability (i.e., the system is not observable in the neighborhood of any agent)  \cite{usman_tsp:07,cattivelli2008diffusion,chen2016dynamic,chen2018resilient,battistelli_cdc,sayedtu12,nuno-suff.ness}. This differs from double time-scale estimation where all necessary information for observability is \textit{directly} communicated to every agent from its neighbors. This requires considerably more communication traffic and information exchange over the network. This implies that the biased (attacked) measurement affects the \textit{residual} (defined as the deviation of the estimated/expected output from the original system output \cite{giraldo2018survey}) at more agents, making it harder to locally \textit{isolate} the faulty sensor.} 
Such additive bias could be, for example, due to false-data injection attacks \cite{guan2017distributed}.
\textcolor{blue}{The general idea in this work is to \textit{locally} detect and isolate such attacks and, further, reconfigure the multi-agent network using substitute measurements to recover (potential) loss of observability.}

\textcolor{blue}{The distributed estimator in this paper performs consensus (on the received data) at the same time-scale of the underlying system (single time-scale), see e.g.,~\cite{khan_cdc:2010,jstsp} for details. We use structured systems theory~\cite{woude:03,jstsp,jstsp14}) to guarantee \textit{generic} or \textit{structural} observability. This helps to partition the system outputs (fed to the agents) into certain observationally-equivalent classes~\cite{icassp2016}. This gives the set of  necessary agents for estimation (whose removal makes the system unobservable) and the set of redundant agents (whose removal results in no observability loss). Subsequently, different strategies are used to substitute the faulty sensor and design inter-agent communications. We propose our attack detection and mitigation strategy based on this specific \textit{agent classification}. In particular, we show that isolation of the attacks related to system rank-deficiency is more challenging and requires certain constrained gain design.  Recall that system rank refers to the rank of the associated matrix to the linear system of differential equations (in the state-space representation), see Section~\ref{sec_system} for more details.}

\textit{Comparison with related literature:}
\textcolor{blue}{
this work develops a \textit{joint} distributed estimation and attack detection/isolation technique, and extends the prior works on resilient distributed estimation subject to \textit{unreliable} sensor measurements~\cite{pereira2013diffusion,spl17} and adversarial attacks~\cite{dutta2019resilient,mitra2019resilient,mustafa2019secure,wen2018distributed,yang2020adversary,su2019finite,li2017sampled,wang2014stochastically,chen2018resilient}. These literature do not detect/isolate the attack, but estimate the system in the presence of (specific) attacks with bounded (steady-state) error, while making simplifying assumptions, e.g., a \textit{noise-free} model. Our work extends \cite{pereira2013diffusion,spl17,dutta2019resilient,mitra2019resilient,mustafa2019secure,wen2018distributed,yang2020adversary,su2019finite,li2017sampled,wang2014stochastically,chen2018resilient} by further considering \textit{distributed/localized} techniques to locate the attacked sensor. 
Further, this work differs from many works on distributed estimation in the literature by relaxing  the observability assumption; for example, \cite{usman_tsp:07,cattivelli2008diffusion,chen2016dynamic,chen2018resilient,battistelli_cdc,sayedtu12,nuno-suff.ness} assume local observability at some (or all) agents.  In contrast, and similar to \cite{he2019secure,he2020secure,flock}, our work makes no such restrictive assumption. However, \cite{he2019secure,he2020secure,flock} perform many iterations of data-fusion (consensus) between two consecutive system steps (\textit{double time-scale} estimation), requiring much faster data-processing/communication rate.}

\textcolor{blue}{In the context of adversarial attacks, most observer-based  detection scenarios assume system and/or measurement noise with bounded support, i.e., they consider an upper bound on the noise variable  \cite{kim2018detection,pajic2015attack,chong2015observability,lee2015secure,shoukry2017secure}. In this paper, we make no such assumption; instead, the noise is assumed to be of infinite support (i.e., it can take any arbitrarily large value with bounded second-order moment). Therefore, we propose probabilistic attack-detection thresholds, in contrast to the deterministic threshold design (or \textit{flag value}) in observer-based detection methods \cite{kim2018detection,pajic2015attack,chong2015observability,lee2015secure,shoukry2017secure}. In another line of research \cite{kailkhura2016data,wang2017data,hashlamoun2017mitigation,chen2016optimal,rosas2017technological,soltanmohammadi2012decentralized,kailkhura2014asymptotic,zheng2017steady}, distributed attack detection without observer/estimator design is considered. These works consider a multi-agent network aiming to detect (typically Byzantine) attack in a sensed signal in a distributed way, with no estimation purpose (due to unknown system model). For example,~\cite{pasqualetti2013attack} uses innovation variance to detect attacks (component malfunctions) in linear-quadratic-Gaussian (LQG) CPS models. However, our main goal is to detect the attacks (in form of biasing anomalies changing the true output values \cite{milovsevivc2017analysis}) deteriorating distributed estimation performance, and, further, to provide a mitigation strategy to restore observability (more precisely, \textit{distributed observability} \cite{globalsip14}). In this regard, this paper performs \textit{simultaneous distributed} estimation and attack-detection, which makes it different from \cite{kailkhura2016data,wang2017data,hashlamoun2017mitigation,chen2016optimal,rosas2017technological,soltanmohammadi2012decentralized,kailkhura2014asymptotic,zheng2017steady,pasqualetti2013attack} performing \textit{only} detection.}

\textcolor{blue}{Of relevance are also \textit{watermarking} strategies \cite{mo2015physical,satchidanandan2016dynamic}, that inject a known input signal (watermark) into the system and track this watermark in the outputs using Chi-square testing ($\chi^2$-detector). Such input injection is not possible for tracking \textit{autonomous} systems, and thus, the physical watermarking is impractical in such cases. 
The distributed strategy in this work is not limited to full-rank LTI systems, in contrast to distributed estimators in \cite{deghat2019detection,battistelli_cdc,sayedtu12,nuno-suff.ness,TCNS2020} over strongly-connected (SC) sensor-networks. Further, unlike the static parameter estimation in \cite{chen2018resilient2} and noiseless centralized attack-detection/estimation in \cite{chen2016dynamic}, this work is based on \textit{distributed} estimation of \textit{noise-corrupted} linear systems. Another relevant topic is compressive sensing  \cite{joseph2018observability,wakin2010observability,li2019distributed,majidi2017distribution,hamidi2016hybrid,xu2015distributed} to translate the data into a compressed dimension, share and  combine the data, reconstruct it to the full dimension, and perform diffusion-based \cite{xu2015distributed} or least mean square (LMS) update \cite{li2019distributed,hamidi2016hybrid,majidi2017distribution} to estimate the original signal. Although the compressed \textit{transmit} of data is applicable in our work (to reduce the communication burden), \textit{distributed dynamic} observability makes our work different from  \cite{li2019distributed,xu2015distributed,majidi2017distribution,hamidi2016hybrid,6712169} based on \textit{static} observability irrespective of the dynamic system model. Recall that this is referred to as the \textit{Static Linear State-Space} (SLS) model in detection literature \cite{giraldo2018survey} and differs from our solution considering \textit{Linear Dynamical State-Space} (LDS) model\footnote{Using the dynamic model of the system (LDS case), fewer outputs are needed to reconstruct the full state of the system (dynamic observability), while in the static or SLS case (with no information of system dynamics) in general more outputs (as many as system states) are needed. Having fewer outputs in the SLS case results in under-determined system of linear equations (unobservability), which mandates substitute recovering solutions such as compressive-sensing or auto-encoder neural networks \cite{agarwal2020assessing}. \textcolor{blue}{A compressive-sensing-based example for the smart-grid application is given in \cite{6712169}, which requires no rank condition on the SLS model.}
}. 
Similarly this work differs from \textit{centralized} estimation in \cite{joseph2018observability,wakin2010observability} with certain assumptions on the sparsity of the initial states \cite{joseph2018observability,wakin2010observability} or system rank  \cite{joseph2018observability}. Autoencoder-based learning is used in some works \cite{wang2020detection,wang2018distributed,wilson2018deep,Khodayar_deep} to distinguish (classify) faulty/attacked data from non-attacked measurement data. In smart-grid applications, the PMU measurements are used to train the detector via either supervised learning \cite{wang2020detection,wang2018distributed} or unsupervised learning \cite{wilson2018deep}. No dynamics is considered in these works (SLS model), 
contrasting our (distributed) observability-based LDS model. Further,  \cite{wang2020detection,wang2018distributed,wilson2018deep} only perform detection with no aim of estimation in the absence of attacks, while some works (see references in \cite{Khodayar_deep}) only perform learning-based estimation with no possibility of detection. Recall that noise (in system dynamics and/or output) plays a key role in the LDS detection. As mentioned before, the assumption on the noise support (finite or infinite) and its value in the finite-case affects the performance of the detection mechanism \cite{kim2018detection,pajic2015attack,chong2015observability,lee2015secure,shoukry2017secure}. Similarly, noise in the output data affects the SLS detection performance, e.g., in power-system applications \cite{li2019distributed,xu2015distributed,majidi2017distribution,hamidi2016hybrid,6712169}. See more details along with a review of centralized physics-based detection mechanisms in \cite{giraldo2018survey}.}

\emph{Main contributions:}
\textcolor{blue}{\begin{inparaenum}[(i)]
\item Our observer-based detection strategy is \textit{localized} and \textit{distributed} over the multi-agent network with \textit{no local observability} assumption at any agent, but  \textit{global observability} at the group of agents. This is key in large-scale, as it enables each agent to detect a (possible) attack on its received output with no central coordination, in contrast to centralized detection scenarios. 
\item Using certain agent classification based on system-rank, we develop detection and attack isolation strategies \textcolor{blue}{which are specific to the measurement types based on the system dynamics (LDS model) (see Section~\ref{sec_necessary} for detailed explanation)}. 
\item The noise is considered over an infinite range with no constraint/bound on its support, which is more realistic for real-world applications (see Remark~\ref{rem_noise}). In this sense, our attack detection and mitigation is categorized as probabilistic (vs. deterministic) thresholding. 
\item In order to prevent repetitive attacks at the same agent by the adversary, we consider an attack mitigation strategy to replace the biased measurement with an observationally-equivalent one (borrowing results from \cite{icassp2016,tnse18}).
\end{inparaenum}
We emphasize that the proposed algorithms for threshold design, agent classification, and mitigation via observational equivalency are of polynomial-order complexity.}


\textit{Notation:} Throughout this paper, scalar and (column) vector variables are respectively represented by lower-case and bold lower-case letters. Further, capital letters represent matrices. The induced $2$-norm of the matrix $A$ is defined as ${\|A\|_2 = \sqrt{\lambda_n}}$ where ${\lambda_n=\rho(A^\top A)}$ and ${\rho(\cdot)}$ denotes the spectral radius of matrix. Further, $|\cdot|$ denotes the Euclidean norm.
 Table~\ref{tab_notation} summarizes the notation in this paper.
\begin{table} [t]
	\centering
		\caption{Notations in the paper (subscript $k$ implies parameter's time index).} 	\label{tab_notation}
		\begin{tabular}{|c|c|}
			\hline
			$n,N$& number of system states, measurements (agents)  \\	
			\hline
			$k$& time index  \\
			\hline
			$\mb{x}_k,\mb{y}_k$ &  column-vector of states, measurements  \\
			\hline
			$\pmb{\nu}_k,\pmb{\zeta}_k$ & zero-mean  system and measurement noise   \\
			\hline
			$\pmb{\tau}_k$ & attack vector  at time $k$  \\	
			\hline
			$A, C$ & system and measurement matrix\\	
			\hline
			$E,R$ & system  and  measurement noise covariance \\			
			\hline
			$\mb{c}_j$ & measurement matrix (column-vector) at agent $j$\\	
			\hline
			$\alpha,\beta,\gamma$ & types of agents \\ \hline
            $\mc{G}_A$ & system digraph associated with system matrix $A$  \\
            \hline
            $\mc{G}_N, \mc{G}_\alpha , \mc{G}_\beta$ & communication network of agents  \\	
            \hline
            $\mc{N}_\alpha(i) , \mc{N}_\beta(i)$ & set of neighbors of agent $i$ over networks $\mc{G}_\alpha , \mc{G}_\beta$  \\	 \hline      $\mc{N}(\cdot,\cdot)$ &  Gaussian distribution  \\  
            \hline
            $\mc{C},\mc{S}^p$ & contraction and parent SCC in the system graph\\
            \hline
            $W$ & stochastic fusion-matrix associated to $\mc{G}_\beta$ \\	
            \hline
            $U$ & adjacency matrix of $\mc{G}_\alpha$  \\
            \hline
            $K$ &  gain matrix (with $K_i$ as $i$-th diagonal-block ) \\
            \hline
            $\widehat{\mb{x}}^i_{k|k-1},\widehat{\mb{x}}^i_{k|k}$ &   \emph{a priori} and \emph{a posteriori} estimate of agent $i$ \\ 
            \hline           
            $\mb{e}_k^i, r_k^i$ &  estimation error and residual of agent $i$ \\
            \hline
            $\theta_\kappa$ &  detection threshold associated with probability $\kappa$ \\
            \hline
             $\kappa$,$\varkappa$ &  threshold probability, probability of false-alarm  \\          
            \hline
            $\mathbf{1}_{NN},\mathbf{1}_N$ & all $1$'s matrix/column-vector of size $N$ \\
            \hline
            $I_{N}$ & Identity matrix of size $N$ \\\hline
            $\mathbb{E}$ &  Expected value operator\\
			\hline
			\hline
	\end{tabular}
\end{table}

\section{Problem Setup}\label{sec_prob}
\subsection{Linear Dynamical System} \label{sec_system}
\textcolor{blue}{Following the discussions in Section~\ref{sec_intro}, we consider \textit{noise-corrupted} linear discrete-time systems (LDS model \cite{giraldo2018survey})}
as, 
\begin{align}\label{eq_sys1}
\mb{x}_{k+1} &= A\mb{x}_k + \pmb{\nu}_k,
\end{align}
with ${\mb{x}_k \in \mathbb{R}^n}$ as the column-vector of states at time $k$, $A$ as the system matrix, 
and ${\pmb{\nu}_k \sim \mc{N}(0,E)}$ as the system noise vector. \textcolor{blue}{Throughout the paper, the system-rank refers to the rank of the system matrix $A$. }
Consider a group of $N$ agents with scalar outputs given by ${y}^i_k = \mb{c}^\top_i\mb{x}_k + \pmb{\zeta}^i_k+{\tau}^i_k$ and the vector form as,
\begin{align}
\mb{y}_k &= C\mb{x}_k + \pmb{\zeta}_k+\pmb{\tau}_k, \label{eq_C}
\end{align}
with ${\mb{y}_k \in \mathbb{R}^N}$ as the column-vector of state measurements (or system outputs) ${\mb{y}_k=(y^1_k,\dots,y^N_k)^\top}$, ${\pmb{\zeta}_k = (\zeta^1_k,\dots,\zeta^N_k)^\top \sim \mc{N}(0,R)}$  as the  measurement  noise vector, and ${\pmb{\tau}_k = (\tau^1_k,\dots,\tau^N_k)^\top}$ as the column-vector of  biasing attack  at the agents. We assume arbitrary attack $\pmb{\tau}_k$ by the adversary, e.g., both fixed stationary attack and non-stationary attacks are considered for simulation (Section~\ref{sec_sim}).
Further, the measurement matrix  ${C=[\mb{c}^\top_1;\dots;\mb{c}^\top_N]}$ is the column concatenation of row-vectors $\mb{c}^\top_i$ associated with agent~$i$ (with~``;'' as  column concatenation). Standard assumptions on Gaussianity and independence of  noise terms  are considered. For example, it is typical to assume that the sensor measurements are independent, making the measurement noise covariance matrix $R$ diagonal.


\begin{rem} \label{rem_noise}
Several papers in the literature (e.g., \cite{kim2018detection,pajic2015attack,chong2015observability,lee2015secure,shoukry2017secure}) assume constrained noise  ${|\pmb{\nu}_k|<\delta}$ and/or ${|\pmb{\zeta}_k|<\delta}$, where the upper bound $\delta$ on the noise support sets the deterministic thresholds for attack detection. For example, in \cite{pajic2015attack} the deterministic threshold at sensor $i$ is defined as ${\mc{D}_i = \|\mc{O}_i\|_2 \|\mb{e}^i_k\|_2 + 2 \delta}$  with ${\|\mc{O}_i\|_2}$  and ${\|\mb{e}^i_k\|_2}$  as the 2-norm of the observability Grammian and the state-estimation error, respectively. In contrast, we make no such finite support assumption (loosely speaking,  ${\delta \rightarrow \infty}$), while it is standard to assume that the second moments of the noise terms are  finite, i.e., ${\mathbb{E}(\pmb{\nu}_k^\top \pmb{\nu}_k) <\infty}$  and ${\mathbb{E}(\pmb{\zeta}_k^\top \pmb{\zeta}_k) <\infty}$.  Assuming unbounded $\delta$, the deterministic threshold, for example $\mc{D}_i$ in \cite{pajic2015attack},  also goes unbounded (${\rightarrow \infty}$), and thus, no attack can be detected. Similar arguments hold for \cite{kim2018detection,chong2015observability,lee2015secure,shoukry2017secure}. 

\end{rem}
\begin{rem} \label{rem_marting}
Note that noise Gaussianity is a standard assumption in most distributed estimation/filtering and  attack detection literature, e.g., see \cite{flock,milovsevivc2017analysis,yang2020distributed,acc13,isj2020,xu2016distributed,shiri2018distributed,dehghani2020deep,cui2012coordinated,rawat2015detection,camsap11,drayer2019detection,pereira2013diffusion,spl17,mustafa2019secure,wen2018distributed,kailkhura2016data,chen2016optimal,zheng2017steady,mo2015physical,satchidanandan2016dynamic,TCNS2020,chen2016dynamic,icassp13,asilomar11,khan2014collaborative}.
\end{rem}

\subsection{Agent Classification based on Structural Analysis}\label{sec_necessary}
The notion of observability used throughout this paper is structural~\cite{asilomar11,woude:03,liu-pnas} and the theory is build on this notion. It is known that the rank deficiency of matrix $A$ and strong-connectivity of \textit{system digraph} $\mc{G}_A$ affect its structural observability properties, and further, its estimation performance. In this direction, using structured systems theory and generic analysis \cite{woude:03,liu-pnas}, we propose specific sensor/agent classification based on the structure (zero-nonzero pattern) of the system matrix $A$ and system digraph $\mc{G}_A$. Using the theory developed in \cite{icassp2016,isj2020}, the agents are partitioned into different classes based on their state-measurements. \textcolor{blue}{We specifically show in Section~\ref{sec_main_alg} that the detection and mitigation logic differs for each class.} First, we describe some relevant graph-theoretic notions. In $\mc{G}_A$, every node represents a state and every link represents a fixed non-zero entry of $A$ ($A_{ij}$ implies ${j \rightarrow i}$  as a link from node $j$ to node $i$). In $\mc{G}_A$ a strongly-connected-component (SCC)  
is a component in which every node is connected to every other node via a path. Define a parent SCC $\mc{S}^p_l$ as an SCC with no out-going links to other SCCs.
Further, a contraction $\mc{C}_l \in \mc{C}$ 
is a component for which ${|\mc{N}_{\mc{G}}(\mc{C}_l)|<|\mc{C}_l|}$,  with ${\mc{N}_{\mc{G}}(\mc{C}_l)=\{b|a \rightarrow b, a \in \mc{C}_l\}}$  
and $|\cdot|$ as the set cardinality.
Based on these graph components, three types of agents are defined as follows,
\begin{itemize}
	\item \textbf{$\alpha$-agent} is an agent with measurement of a state node in a contraction $\mc{C}_l$.
	\item \textbf{$\beta$-agent} is an agent with measurement of a state node in a parent SCC $\mc{S}^p_l$.
	\item \textbf{$\gamma$-agent} is any agent which is neither type $\alpha$ nor $\beta$.
\end{itemize}
\begin{figure}[t]
	\centering
	\includegraphics[width=3.44in]{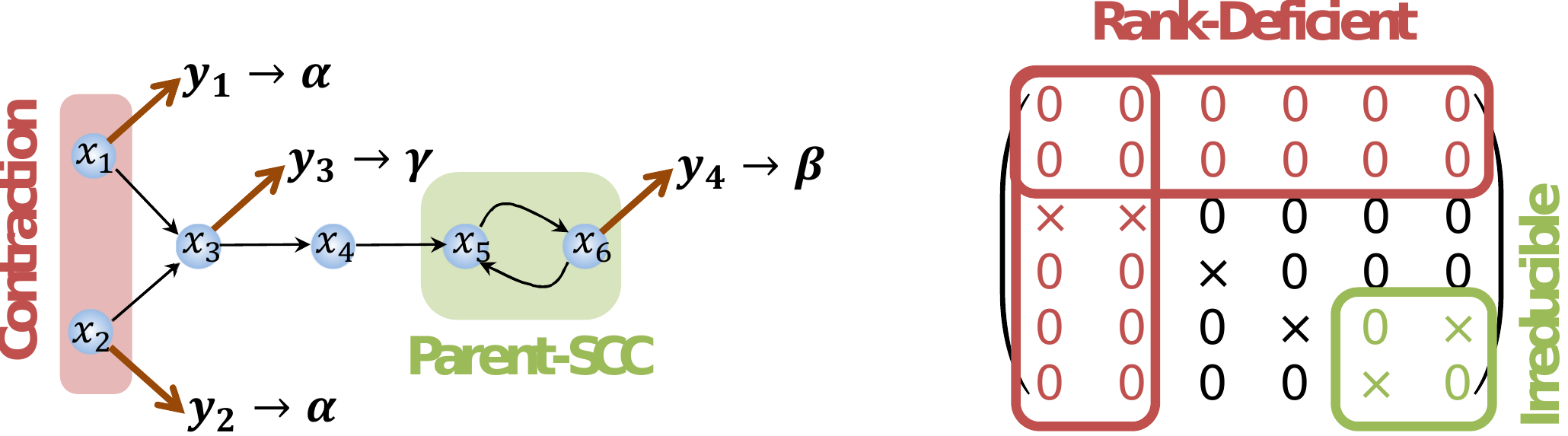}
	\caption{ \textcolor{blue}{This figure illustrates the proposed agent/measurement classification over a simple system digraph and its associated system matrix: $\alpha$-agents with outputs $y_1$ and $y_2$ from a contraction (two state nodes contracting into one state node), $\beta$-agent  with output $y_4$ from a parent SCC (two linked state nodes with no outgoing link to other components), and a redundant $\gamma$-agent (with output $y_3$) which is neither type $\alpha$ nor type $\beta$. As illustrated in the zero-nonzero pattern of the system matrix (right), the contraction represents system (structural) rank-deficiency and the parent-SCC is associated with the irreducible block (with zero entries in the upper/lower non-diagonal blocks). See more details in \cite{icassp2016}.  }  } 
	\label{fig_prob}
\end{figure}
An example of such classification is given in Section~\ref{sec_sim}. This partitioning has two advantages: (i) it allows using a different communication topology for different types of agents and simpler topology design when one or the other type of agents is not present; and (ii) it allows for the attack detection and mitigation strategy to be specifically defined for each type (see details in Section~\ref{sec_main_alg}). In particular, following \cite{icassp13}, it can be shown that any $\alpha$-agent recovers the (structural) rank condition for observability, while the $\beta$-agent recovers the output-connectivity of the system digraph \cite{asilomar11}. Therefore, both $\alpha$ and $\beta$-agents are necessary for observability, while removing (redundant) $\gamma$-agents has no effect on system observability. \textcolor{blue}{Recall that the structural properties are irrespective of the numerical values of system parameters \cite{woude:03}; therefore, for a structure-invariant matrix $A$ the proposed classification is fixed and time-invariant.}

\subsection{Problem Statement}
\begin{figure}[t]
	\centering
	\includegraphics[width=2.8in]{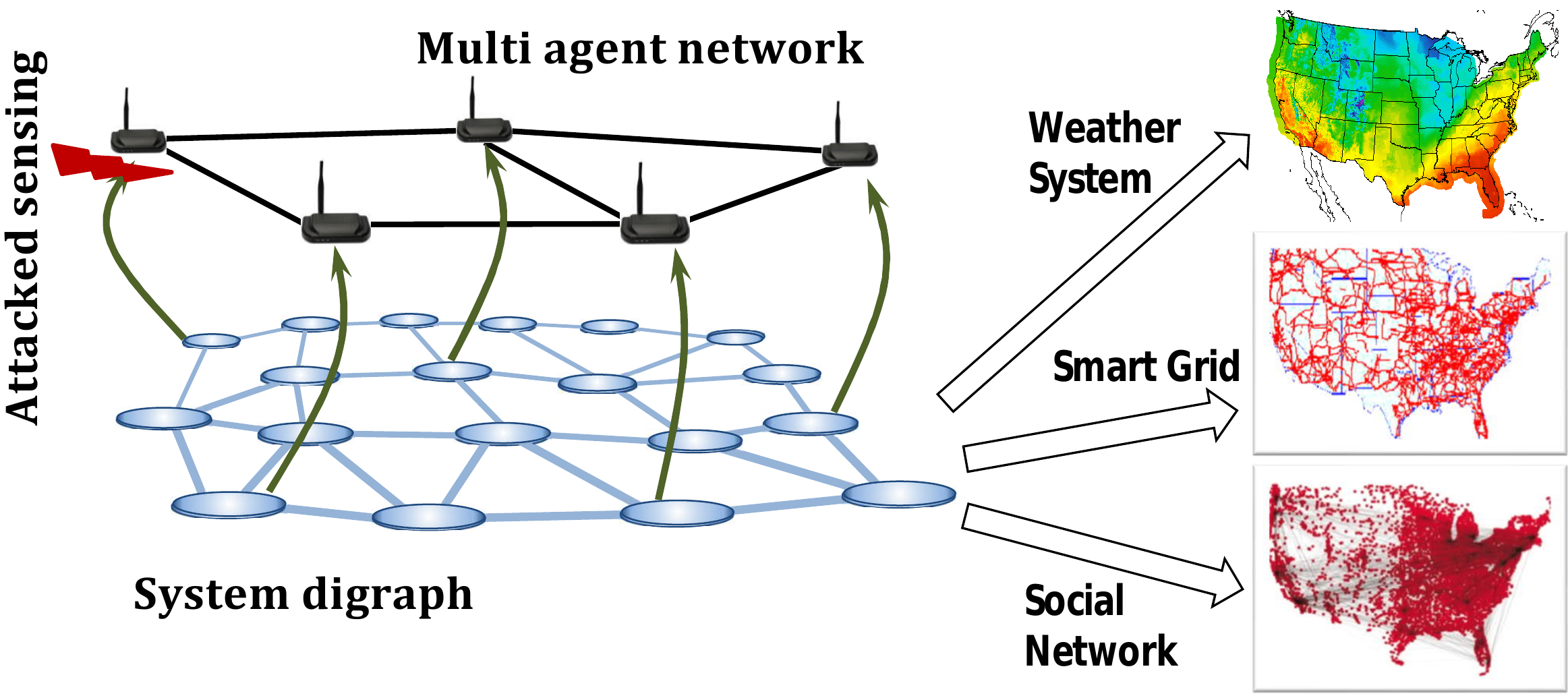}
	\caption{ In this work, there exist two graph representations: (i) system digraph $\mathcal{G}_A$ (see Section~\ref{sec_necessary}), representing the interactions of system states (or system nodes), and (ii) multi-agent network (denoted by ${\mc{G}_N = \mc{G}_\alpha \cup \mc{G}_\beta}$). The system digraph models a large-scale state-space system, e.g., social network, power grid, or weather system.  The green arrows show the state measurements/outputs which vulnerable to (possible) adversarial attacks. The agents/sensors are classified based on  their state measurements from specific components in $\mathcal{G}_A$  (see examples in Section~\ref{sec_sim}) and track the global system state locally (i.e., performing distributed estimation) via sharing information  over the network ${\mc{G}_N}$. Attacked (biased)  measurements may affect the estimation performance at all agents. The proposed algorithm in this work enables each agent to \textit{locally} detect if its measurement/output is attacked  or not, and further provides mitigation techniques for resilient estimation. }
	\label{fig_prob}
\end{figure}
This paper considers a group of sensors/agents taking \textit{noise-corrupted} measurements in the form \eqref{eq_C} of a dynamical system (e.g., social network or power grid) in the form \eqref{eq_sys1} represented by a \textit{system digraph}~$\mathcal{G}_A$, see Fig.~\ref{fig_prob}. The agents perform distributed estimation over a network, denoted by ${\mc{G}_N = \mc{G}_\alpha \cup \mc{G}_\beta}$ 
to track the state of the noisy dynamical system \eqref{eq_sys1}. \textcolor{blue}{Note that the networks $\mc{G}_\alpha$, $\mc{G}_\beta$, and their union $\mc{G}_N$ include all the agents of type $\alpha$, $\beta$, and $\gamma$.} It is assumed that an adversarial attacker aims to add an arbitrary value ${\tau}^i_k$ (at any time $k$) to make the measurement at (one or more) agent $i$ biased from its original value. Since the dynamical system is not necessarily observable at any agent, the biased measurements (at $\alpha/\beta$-agents) affect the estimation error at all agents and result in the degradation of the distributed estimation performance. The problem here is to find a strategy to detect (and isolate) such instantaneous attacks \textit{locally at each agent}.
In particular, we propose a probabilistic detection strategy that returns the probability of attack (at each agent), instead of deterministic strategies returning 0-1 (NoAttack-Attack).  The next question addressed in this paper is  how to recover the potential loss of observability due to removing the attacked measurement depending on its type ($\alpha$, $\beta$, or $\gamma$). Such countermeasures prevent the same adversarial attack by removing the attacked agent/measurement. As explained in Section~\ref{sec_main_alg}, the attacked measurement can be replaced with a new \textit{observationally-equivalent} one to avoid possible repetitive attacks at the same agent.

\subsection{Assumptions} \label{sec_ass}
\begin{enumerate} [(i)]
	\item The pair $(A,C)$ is observable. The pairs ${(A,\mb{c}^\top_j)}$  and ${(A, \mb{c}^\top_{\mc{N}_\alpha(j)})}$  are not necessarily observable at any sensor $j$ or in its neighborhood denoted by ${\mathcal{N}_\alpha(j) \cup \mathcal{N}_\beta(j)}$  (see details in  Section~\ref{sec_est}). This implies that the underlying system $A$ is not necessarily  observable in the neighborhood of any agent.
	\item The noise terms $\pmb{\nu}_k$, $\pmb{\zeta}_k$ are iid Gaussian,
	see Remark \ref{rem_noise}.
	\item The known system matrix $A$ is not necessarily stable, i.e., its spectral radius $\rho({A})$ can be potentially greater than $1$. In other words, this paper applies to both stable and unstable systems. 
	\item The adversary
	can manipulate the state measurements at a subset of sensors by adding erroneous additive term $\pmb{\tau}_k$ at any time $k$. For example,   ${\tau}^i_k$ can be from a uniform distribution over $[-l_\tau,l_\tau]$ with ${l_\tau \gg \|R\|_2,~l_\tau \gg\|E\|_2}$ (${l_\tau \rightarrow \infty}$ in general) or ${\tau}^i_k$ can be a fixed value. In general, the term ${\tau}^i_k$ may be non-zero at some time-instants $k$ (instantaneous attack) and zero at some other times.
\end{enumerate}

\section{Distributed Estimation under Possible Measurement Attacks} \label{sec_est}
In this section, we propose a consensus-based distributed estimation (filtering) protocol over the multi-agent network. The proposed protocol performs \textit{one} iteration of information sharing and consensus  between every two consecutive steps of system dynamics as follows:

\begin{align}\label{eq_p}
\widehat{\mb{x}}^i_{k|k-1} &= \sum_{j\in\mathcal{N}_\beta(i)} W_{ij}A\widehat{\mb{x}}^j_{k-1|k-1},
\\\label{eq_m}
\widehat{\mb{x}}^i_{k|k} &= \widehat{\mb{x}}^i_{k|k-1} + K_i \sum_{j\in \mc{N}_\alpha(i)}\mb{c}_j \left({y}^j_k-\mb{c}^\top_j\widehat{\mb{x}}^i_{k|k-1}\right),
\end{align}
where ${y}^j_k$ is the measurement of agent $j$ at time $k$ that could be attack-corrupted (or biased), $\mathcal{N}_\beta(i)$ and $\mathcal{N}_\alpha(i)$ are the neighborhood of agent $i$, respectively, over network $\mc{G}_\beta$ and $\mc{G}_\alpha$,   $K_i$ is the local feedback gain (or the observer gain) matrix at agent $i$, and $\widehat{\mb{x}}^i_{k|k-1}$ and $\widehat{\mb{x}}^i_{k|k}$ are the (column-vector of) estimates of system state $\mb{x}_k$ at agent $i$ given the measurements, respectively, up to time $k-1$ and  $k$. In fact, $\widehat{\mb{x}}^i_{k|k-1}$ is the \textit{a-priori} estimate (or \textit{prediction}) and $\widehat{\mb{x}}^i_{k|k}$ is the \textit{posteriori} estimate after \textit{measurement-update} at time-step $k$.
\begin{rem}
	In this work, the combination of the following two graphs forms the multi-agent network: (i) $\mc{G}_\beta$ over which agents share the estimates $\widehat{\mb{x}}^j_{k-1|k-1}$, and (ii) $\mc{G}_\alpha$ over which agents share their measurements ${y}^j_k$. Define matrices $W$ and $U$ as the associated matrices to the graphs $\mc{G}_\beta$ and $\mc{G}_\alpha$, respectively. 
	The matrix $U=\{U_{ij}\}$ is the 0--1 adjacency matrix of $\mc{G}_\alpha$, with $U_{ij}=1$ associated to the link ${j\rightarrow i}$  in $\mc{G}_\alpha$ from $\alpha$-agent $j$ to every agent $i$. The \textit{non-zero} entries of $W=\{W_{ij}\}$ take values in the range ${0<W_{ij}\leq 1}$  associated to the link ${j\rightarrow i}$  in $\mc{G}_\beta$.  
\end{rem}
Matrix $W$  is  row-stochastic  to ensure \textit{consensus} on a-priori estimates, i.e., $\sum_{j =1}^n W_{ij} = \sum_{j\in\mathcal{N}_\beta(i)} W_{ij} = 1$ for all $i,j$. Such a matrix $W$ (and the graph $\mc{G}_\beta$) can be formed via  distributed algorithms in \cite{themis_stochastic}.
The structure of $\mc{G}_\beta$ and $\mc{G}_\alpha$ (and the associated matrices) need to be designed properly for bounded steady-state  estimation error, see Section~\ref{sec_err_stab}.  
\begin{rem} \label{rem_scale}
	 The proposed protocol \eqref{eq_p}-\eqref{eq_m} is a single time-scale distributed estimator, where the estimation is performed at the same time-scale of the system dynamics. This is in contrast to the double time-scale protocols  \cite{he2019secure,he2020secure,flock}, which require much faster estimation and communication rate than the sampling rate of the system dynamics, and, therefore, demand more costly communication and processing equipment. However, the observability assumption in \cite{he2019secure,he2020secure,flock} is similar to  Assumption~(ii), which makes such scenarios suitable for large-scale applications as the proposed protocol \eqref{eq_p}-\eqref{eq_m}; see examples in Section~\ref{sec_sim}. 
\end{rem}


\vspace{\belowdisplayskip}
\par\noindent\rule{\dimexpr(0.5\textwidth-0.5\columnsep-0.4pt)}{0.4pt}%
\rule{0.4pt}{6pt}
\begin{strip}
Denote the estimation error at agent $i$ at time $k$ by ${\mb{e}_{k}^i \triangleq \mb{x}_{k|k} - \widehat{\mb{x}}^i_{k|k}}$ and let ${\mb{e}_{k} = (\mb{e}_{k}^{1}; \dots;\mb{e}_{k}^N)}$  be the global or collective error. Then, the following proposition defines the error dynamics of the protocol \eqref{eq_p}-\eqref{eq_m}. 
\begin{prop}
The global error dynamics for protocol~\eqref{eq_p}-\eqref{eq_m} is,
\begin{align}
\mb{e}_{k} &= (W\otimes A - KD_C(W\otimes A))\mb{e}_{k-1} +
\pmb{\eta}_k  
= \widehat{A}\mb{e}_{k-1} +
\pmb{\eta}_k, \label{eq_err1} \\
\pmb{\eta}_k &= \mathbf{1}_N \otimes \pmb{\nu}_{k-1} 
- K D_C(\mathbf{1}_N \otimes \pmb{\nu}_{k-1}) - K\overline{D}_C\pmb{\zeta}_{k} -K\overline{D}_C\pmb{\tau}_{k},
\label{eq_eta}
\end{align}
where $\pmb{\eta}_k$ collects the noise terms, ${\widehat{A}\coloneqq W\otimes A - KD_C(W\otimes A)}$,
${K\triangleq \mbox{blockdiag}[K_i]}$,   ${D_C\triangleq\mbox{blockdiag}[\sum_{j\in \mathcal{N}_\alpha(i)} \mb{c}_j \mb{c}^\top_j]}$, and  $\overline{D}_C \triangleq (U \otimes \mb{1}_n) \circ (\mb{1}_N \otimes C^\top  ) $ with ``$\circ$'' and ``$\otimes$'', respectively, as the entrywise (Hadamard) and Kronecker product.
\end{prop}
\begin{proof}
The error at each agent $i$ is as follows,
\begin{align}
\mb{e}_{k}^i &=\mb{x}_{k} - \Bigl(\sum_{j\in\mathcal{N}_\beta(i)} W_{ij}A\widehat{\mb{x}}^j_{k-1|k-1} \nonumber + K_i \sum_{j\in \mc{N}_\alpha (i)}\mb{c}_j
({y}^j_k-\mb{c}^\top_j\sum_{j\in\mc{N}_\beta (i)} W_{ij}A\widehat{\mb{x}}^j_{k-1|k-1})\Bigr). \nonumber
\end{align}
Recalling stochasticity of $W$ matrix, we have ${A\mb{x}_{k-1}=\sum_{j\in \mathcal{N}_\beta(i)} W_{ij}A\mb{x}_{k-1}}$. Substituting this along with equations~\eqref{eq_sys1}-\eqref{eq_C},
\begin{align}
\mb{e}_{k}^i &=\sum_{j\in \mathcal{N}_\beta(i)} W_{ij}A\mb{x}_{k-1} - \sum_{j\in\mathcal{N}_\beta(i)} W_{ij}A\widehat{\mb{x}}^j_{k-1|k-1} \nonumber \\ &- K_i \sum_{j\in \mathcal{N}_\alpha(i)} \mb{c}_j \mb{c}^\top_j\Bigl(\sum_{j\in \mathcal{N}_\beta(i)}  W_{ij}A\mb{x}_{k-1} - \sum_{j\in \mathcal{N}_\beta(i)}   W_{ij}A\widehat{\mb{x}}^j_{k-1|k-1} \Bigr) + \pmb{\nu}_{k-1}-K_i\sum_{j\in \mathcal{N}_\alpha(i)} \mb{c}_j \zeta^j_k +
\mb{c}_j \tau^j_{k}+ \mb{c}_j \mb{c}^\top_j\pmb{\nu}_{k-1} \nonumber \\
&=\sum_{j\in \mathcal{N}_\beta(i)} W_{ij}A\mb{e}_{k-1}^j -K_i \sum_{j\in \mathcal{N}_\alpha(i)} \mb{c}_j \mb{c}^\top_j\sum_{j\in \mathcal{N}_\beta(i)}  W_{ij}A\mb{e}^j_{k-1}+ \pmb{\eta}_{k}^i, \label{eq_err_i}
\end{align}
with ${\pmb{\eta}_{k}^i \triangleq \pmb{\nu}_{k-1}-K_i\sum_{j\in \mathcal{N}_\alpha(i)} (\mb{c}_j \zeta^j_k +
\mb{c}_j \tau^j_{k}+ \mb{c}_j \mb{c}^\top_j\pmb{\nu}_{k-1} )}$. Using the definition of Kronecker and entrywise products, the collective error and noise term follow Eq.~\eqref{eq_err1}-\eqref{eq_eta}. 
\end{proof}
\end{strip}
\hfill\rule[-6pt]{0.4pt}{6.4pt}%
\rule{\dimexpr(0.5\textwidth-0.5\columnsep-1pt)}{0.4pt}

\subsection{Error Stability} \label{sec_err_stab}
The following lemma establishes the stability condition of the error dynamics~\eqref{eq_err1}-\eqref{eq_eta}. 
\begin{lem}
	The necessary condition for error dynamics~\eqref{eq_err1}-\eqref{eq_eta} to be stable is that the pair ${(W\otimes A,D_C)}$  is observable. 
\end{lem}
\begin{proof}
	The proof follows the Kalman stability theorem on the error dynamics \eqref{eq_eta}. More information can be found in \cite{bay,asilomar11,usman_cdc:11} on error stability of linear observer design.
\end{proof}

Note that ${(W\otimes A,D_C)}$-observability is also referred to as the distributed observability \cite{globalsip14}. Using structured system theory (generic analysis), distributed observability can be formulated as the observability of the Kronecker product of the graphs $\mc{G}_A$ and $\mc{G}_\beta$. Following the observability analysis of  Kronecker composite networks in  \cite{kronecker_TSIPN}, the following lemma determines the sufficient connectivity of $\mathcal{G}_\beta$ and $\mathcal{G}_\alpha$. 

\begin{lem} \label{lem_kron}
	The pair ${(W\otimes A,D_C)}$  is observable if and only if the following conditions hold: 
	\begin{enumerate}
		\item $\mathcal{G}_\beta$ is strongly-connected (SC) with self-link at each agent, which further implies that $W$ is irreducible.
		\item $\mathcal{G}_\alpha$ is a hub-network in which every $\alpha$-agent is a hub, i.e., there is a directed link from every $\alpha$-agent  to every other agent in $\mathcal{G}_\alpha$. Further,  ${i \in \mathcal{N}_\alpha(i)}$ for every agent $i$.
	\end{enumerate}
\end{lem}
\begin{proof}
	We provide the sketch of the proof here and refer the interested reader to \cite{kronecker_TSIPN} for more details. For (structural) observability two conditions on the associated composite graph need to be satisfied \cite{asilomar11,liu-pnas}: (i) the output connectivity condition, implying the existence of a directed path from every state node in the system graph $\mathcal{G}_A$ to an agent (output), and (ii) the rank condition, implying a direct output of (at least) one state node in every contraction in $\mathcal{G}_A$ for system-output rank recovery. In this work, the global system graph associated with ${W\otimes A}$  is the Kronecker-product of $\mathcal{G}_A$ and $\mathcal{G}_\beta$. Recall that for  ${(W\otimes A,D_C)}$-observability (or distributed observability) the global system state must be observable to every agent. Therefore, to satisfy condition (i), every state node needs to be connected via a directed path to every agent, which justifies strong-connectivity of $\mathcal{G}_\beta$. On the other hand, to satisfy condition (ii), the outputs from state nodes measured by all $\alpha$-agents (including one node in every contraction) need to be directly shared among all agents to recover their system-output rank. This implies that for any $\alpha$-agent $j$, we have ${j \in \mathcal{N}_\alpha(i), \forall i\in\{1,\dots,N\} }$. This justifies the connectivity of  $\mathcal{G}_\alpha$, and completes the proof.
\end{proof}
With $\mathcal{G}_\beta$ and $\mathcal{G}_\alpha$ satisfying the conditions in Lemma~\ref{lem_kron}, the block-diagonal gain matrix $K$ can be designed such that ${\rho(\widehat{A})<1}$, i.e., $\widehat{A}$ is a Schur matrix. In fact, the gain matrix $K$ is known to be  the solution to the  Linear-Matrix-Inequality (LMI) ${X-\widehat{A}^\top X\widehat{A} \succ 0}$
or equivalently,
\begin{align}\label{eq_lmi}
\left(
\begin{array}{cc}
X &  \widehat{A}^\top X \\
X\widehat{A}& X
\end{array}
\right) \succ 0,
\end{align}
for some ${X\succ0}$ (where ``$\succ$'' denotes positive-definiteness). However, to satisfy the distributed condition, $K$ needs to be further block-diagonal in order to satisfy information locality. Following \cite{rami:97,usman_cdc:11}, 
iterative cone-complementarity optimization
method is adopted to design the proper $K$ matrix with polynomial-order complexity. Applying such $K$ matrix, we have ${\rho(\widehat{A})<1}$, which implies stability and steady-state boundedness of the error in the attack-free case.

\subsection{Performance Analysis in the Attack-free Case} \label{sec_perf}
Next, we provide the performance analysis of the proposed distributed estimator (filter) \eqref{eq_p}-\eqref{eq_m} in the attack-free case. Following the same analogy as in \cite{usman_tsp:07,cattivelli2008diffusion,flock,khan2014collaborative}, we analyze the mean performance and mean-square performance of the protocol \eqref{eq_p}-\eqref{eq_m} for $\pmb{\tau}_k = \mb{0}$. 
\begin{lem}\label{lem_Ee}
	Let ${\mb{e}_{\infty} \triangleq \lim_{k \rightarrow \infty} \mb{e}_{k}}$  denote the steady-state error of the proposed estimator \eqref{eq_p}-\eqref{eq_m}. Then, ${\mathbb{E}(\mb{e}_{\infty})=\mb{0}}$.
\end{lem}
\begin{proof}
    Taking expectation  of the error dynamics~\eqref{eq_err1},
    \begin{align} \label{eq_err_E}
    	\mathbb{E}(\mb{e}_{k}) &= \widehat{A}\mathbb{E}(\mb{e}_{k-1}) +
    	\mathbb{E}(\pmb{\eta}_k).
    \end{align}
    Recall from Section~\ref{sec_err_stab} that ${\rho(\widehat{A})<1}$ and following from \cite{usman_tsp:07,khan2014collaborative}, it is clear that the first term in \eqref{eq_err_E} vanishes asymptotically.   Then, from \eqref{eq_eta} in the attack-free case (${\pmb{\tau}_k = \mb{0}}$),
    \begin{align} \nonumber
    	\mathbb{E}(\mb{e}_{\infty}) &= \mathbb{E}(\pmb{\eta}_{\infty})\\ \nonumber
    	&= \mathbf{1}_N \otimes \mathbb{E}(\pmb{\nu}_{\infty}) 
    	- K D_C(\mathbf{1}_N \otimes \mathbb{E}(\pmb{\nu}_{\infty})) - K\overline{D}_C\mathbb{E}(\pmb{\zeta}_{\infty}).
    \end{align}
    Recall from Section~\ref{sec_system} that ${\mathbb{E}(\pmb{\nu}_{k})=\mb{0}}$  and ${\mathbb{E}(\pmb{\zeta}_{k}) = \mb{0}}$. This implies that ${\mathbb{E}(\mb{e}_{\infty})= \mb{0}}$  and the lemma follows.		
\end{proof}
\begin{lem} \label{lem_Ee2}
	Define ${Q_k \coloneqq \mathbb{E}(\mb{e}_k\mb{e}_k^{\top})}$  and $\Phi \coloneqq \mathbb{E}(\pmb{\eta}_k\pmb{\eta}^\top_k)$. Let ${Q_{\infty} = \lim_{k \rightarrow \infty} Q_{k}}$  denote the collective error covariance at the steady-state. For error dynamics \eqref{eq_err1} in the attack-free case,
	\begin{align} \label{eq_pinfty}
		\|Q_{\infty}\|_2 \leq \frac{a_1N\|E\|_2+a_2 \|\overline{R}\|_2}{1-b^2},
	\end{align}
	with
	${a_1 \triangleq \|I_{Nn}- K D_C\|_2^2}$, and ${a_2 \triangleq \|K\|_2^2}$,  ${\overline{R}\triangleq\mbox{blockdiag}[\sum_{j\in \mathcal{N}_\alpha(i)} \mb{c}_j R_{jj} \mb{c}^\top_j]}$.	
\end{lem}
\begin{proof}
	Following \cite{khan2014collaborative} with ${ \|\widehat{A}\|_2 \triangleq b }$,
	\begin{align} \label{eq_pinfty1}
		\|Q_{\infty}\|_2  \leq \frac{\|\Phi\|_2}{1-b^2}.
	\end{align}
	From \eqref{eq_eta} we have,
	\begin{align} \nonumber
		\pmb{\eta}_k\pmb{\eta}^\top_k &= (I_{Nn}- K D_C)(\mathbf{1}_{NN} \otimes \pmb{\nu}_{k-1}\pmb{\nu}_{k-1}^\top)(I_{Nn}- K D_C)^\top\\ &+ (K\overline{D}_C) \pmb{\zeta}_k\pmb{\zeta}_k^\top (K\overline{D}_C)^\top.
	\end{align}
	Then, from~\eqref{eq_eta},
	\begin{align} \nonumber
		\|\Phi\|_2 &\leq \|(I_{Nn}- K D_C)(\mathbf{1}_{NN} \otimes E)(I_{Nn}- K D_C)^\top\|_2\\ \nonumber
		&+ \|(K\overline{D}_C) R (K\overline{D}_C)^\top\|_2.
	\end{align}
	Using the fact that ${\|\mathbf{1}_{NN} \otimes E\|_2=N\|E\|_2}$,
	\begin{align}
		\|\Phi\|_2 \leq \|I_{Nn}- K D_C\|_2^2 N \|E\|_2 + \|K\|_2^2 \|\overline{R}\|_2,
	\end{align}
	and applying equation \eqref{eq_pinfty1} results in \eqref{eq_pinfty}.
\end{proof}
In fact, Lemma~\ref{lem_Ee} implies that the estimator \eqref{eq_p}-\eqref{eq_m} is unbiased in the absence of attacks, while Lemma~\ref{lem_Ee2} states that its mean-square estimation error (also known as mean-square deviation \cite{cattivelli2008diffusion}) is bounded in steady-state.

\section{Main Algorithm}\label{sec_attack}
We now describe the attack detection logic. 
Define the residual at every agent $i$  as the absolute difference value  between the original output ${y}_k^i$ and the estimated output, 
\begin{align} 
r_k^i &\triangleq |{y}_k^i-\widehat{{y}}_{k}^i|
= |\mb{c}^\top_i \widehat{A}_i \mb{e}_{k-1}+\mb{c}^\top_i \pmb{\eta}_k^i+{\zeta}^i_k+\tau^i_k|.
\label{eq_r}
\end{align}
Note that the residual defined above based on the absolute-value is a standard definition, which is irrespective of the attack being positive (${\tau^i_{k}>0}$) or negative (${\tau^i_{k}<0}$) and works for both sign-preserving and sign-changing attacks. As shown in Lemmas~\ref{lem_Ee} and~\ref{lem_Ee2}, in the attack-free case with ${\tau^i_{k}=0}$, the estimation error $\mb{e}_k^i$, and therefore, the residual $r_k^i$ is bounded steady-state stable and unbiased at all agents. Note that in general ${\widehat{A}_i \mb{e}_{k-1} \rightarrow 0}$  due to Schur stability of $\widehat{A}$, while the second term in \eqref{eq_r} is,
\begin{align} 
\mb{c}^\top_i\pmb{\eta}_{k}^i =  \mb{c}^\top_i\pmb{\nu}_{k-1} -\mb{c}^\top_iK_i\sum_{j\in \mathcal{N}_\alpha(i)} \Bigl(\mb{c}_j \zeta^j_k +
\mb{c}_j {\tau}^j_{k}+ \mb{c}_j \mb{c}^\top_j\pmb{\nu}_{k-1} \Bigr). \label{eq_r_eta}
\end{align}
In case of an attack on  agent $i$, i.e.,~${\tau^i_{k}\neq0}$, the term $\mb{c}^\top_i\pmb{\eta}_{k}^i$ is biased at agent~$i$. This biased residual can be used to find (isolate) the attacked agent. In this sense,  first, we need to define a threshold on the residuals to distinguish the effect of noise terms (in absence of attacks) and the biasing attacks.

\subsection{Probabilistic Threshold Design}
Here, the probabilistic detection thresholds  are defined based on ${ Q_{\infty}}$ in \eqref{eq_pinfty}. For each agent define, 
	\begin{align} \label{eq_theta1}
	 \frac{	\|Q_{\infty}\|_2 }{N} \leq \frac{a_1N\|E\|_2+a_2 \|\overline{R}\|_2}{N(1-b^2)}  \teL \Theta_1.
\end{align}
Then, for specific false alarm rates  and attack detection probabilities $\kappa$, one can consider different detection-levels ${m \in \mathbb{R}_{>0}}$ as described in Fig.~\ref{fig_normal_dist}. A detection-level $m$ represents a specific probability threshold $\kappa$ associated with the Gaussian PDF of the estimation error in the attack-free case.
Then,  the thresholds $\theta_\kappa$ are designed as follows.
\begin{lem} \label{lem_threshold}
Following the assumptions in Section~\ref{sec_ass},  given the noise covariance $R$ and $E$ and the residuals $ r_k^i$ from Eq. \eqref{eq_r}, the attack detection threshold for a  detection-level ${m \in \mathbb{R}_{>0}}$ is,
\begin{align} \label{eq_thresold}
\theta_\kappa \coloneqq m\Theta_2^i,~\Theta_2^i \coloneqq |\mb{c}^\top_i|\Theta_1 + R_{ii}
\end{align}
where ${\kappa = \mbox{erf}(\frac{m}{\sqrt{2}})}$
is detection probability (with $\mbox{erf}(\cdot)$ as the \textit{Gauss error function}), $\mb{c}_i$ is the measurement column-vector at agent $i$, and $\Theta_1$ follows \eqref{eq_theta1}.
\end{lem}
\begin{proof}
The proof directly follows from Lemma~\ref{lem_Ee} and~\ref{lem_Ee2} and the results in \cite{khan2014collaborative}. 
From Lemma~\ref{lem_Ee} and~\ref{lem_Ee2}, ${\mathbb{E}(\mb{e}_k^i)=\mb{0}}$  for attack-free case, and following the zero-mean Gaussian distribution of the noise terms in $\pmb{\eta}_k$ (including ${\pmb{\nu}_{k}}$  and $\pmb{\zeta}_k$) and linearity of the  error dynamics~\eqref{eq_err1}-\eqref{eq_eta} and the protocol \eqref{eq_p}-\eqref{eq_m}, it is straightforward to see that $\mb{e}_k^i$ and $r_k^i$ are Gaussian; see details in \cite{khan2014collaborative}.
Then, from standard textbooks on Gaussian distribution (e.g., \cite{krishnamoorthy2016handbook}) 
and Eq.~\eqref{eq_r} in attack-free case, the probability of ${|r_k^i|\leq m \Theta_2^i}$  with ${\Theta_2^i = |\mb{c}^\top_i|\Theta_1+R_{ii}}$  is determined via the value of the normal deviate less than $m\Theta_2^i$, i.e., ${\kappa = \mbox{erf}(\frac{m}{\sqrt{2}})}$. Recall that $\Theta_2^i$ is the residual variance and $R_{ii}$ is the measurement noise variance at agent $i$. Then, in presence of attack, both error $\mb{e}_k^i$ and residual $r_k^i$   are biased by some products of $\tau_k^i \neq 0$ (due to linearity). 
In this case, the residual follows a biased Gaussian distribution with non-zero mean. Following statistical hypothesis testing for the two Gaussian distributions with equal variance (assuming equally likely a-priori hypothesis), if the residual $|r_k^i|$  is greater than $m\Theta_2^i$ then the probability of attack is ${\kappa}$ and probability of false alarm is ${1-\kappa}$. This justifies the probability thresholds  ${\theta_\kappa \triangleq m\Theta_2^i}$ (as illustrated in Fig.~\ref{fig_normal_dist}) and completes the proof. 
\end{proof}
\begin{figure}[t]
	\centering
	\includegraphics[width=3.1in]{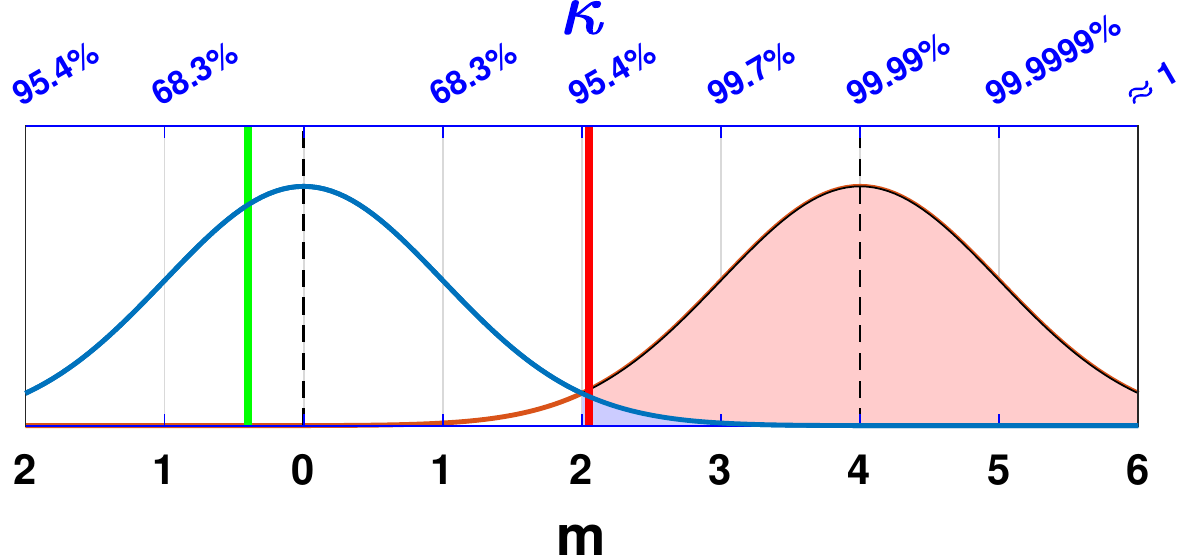} \vspace{-0.3cm}
	\caption{This figure illustrates the attack detection logic in Lemma~\ref{lem_threshold}. The confidence intervals for the normalized residual in the absence of attack (blue curve) are shown. Each value of $m$ in Eq. \eqref{eq_thresold} associated with a confidence interval represents a probability threshold $\kappa$ associated with the Gaussian PDF of the residual.  As an example, the red and green lines represent two  normalized residual values $\frac{r_k}{\Theta_2}$ via Eq.~\eqref{eq_r} and \eqref{eq_thresold}. Following the binary hypothesis testing (maximum-likelihood case), the threshold on the residual is the intersection (midpoint) of the two PDFs, where the residual belongs to the PDF in the presence of attack (red curve). Since the residual is over the threshold $\theta_\kappa$ with $m=2$ (${{r_k}>2{\Theta_2}}$), probability of attack is more than ${\kappa = 95.4\%}$. This probability is equal to the red shaded area (since both PDFs follow the same normal distribution), while the probability of false alarm ($1-\kappa=4.6\%$) is shaded by blue. Clearly, this gives the highest probability of detection, while for higher threshold values (larger $\kappa$) the residual is not detected as biased/attacked. For the green residual with ${|r_k|<{\Theta_2}}$, the residual is most likely due to (system/measurement) noise, which is also evident from the (blue) PDF. Recall that, in general, $m$ may take (positive) real values over infinite range ($m \rightarrow \infty$).}  
	\label{fig_normal_dist}
\end{figure}
The parameter $m$ in \eqref{eq_thresold} and Lemma~\ref{lem_threshold} can take any real (or integer) value in $\mathbb{R}_{>0}$. Some typical threshold probability values $\kappa$  for integer values of $m$ are given in Table~\ref{tab_kappa}. Clearly, higher values of $m$ (and $\kappa$) implies lower false alarm rates.
\begin{table} [h] 
		\centering
		\caption{Different threshold probabilities $\kappa$ for integer $m$ in Eq. \eqref{eq_thresold}. }
		\label{tab_kappa}
		\begin{tabular}{|c|c|c|c|c|} 
			\hline
			$\frac{m}{2}$& $1$ & $2$ & $3$ & $4$  \\
			\hline
			Threshold probability $\kappa$ &  $68.3\%$ & $95.4\%$ & $99.7\%$  & $99.99\%$ \\
			\hline
			\hline
		\end{tabular}
\end{table}
\begin{rem} \label{rem_erf}
   A straightforward sequel to Lemma~\ref{lem_threshold} is that one can design the threshold $\theta_\kappa$ for a given false-alarm rate $\varkappa = 1-\kappa$ as ${\theta_\kappa = \sqrt{2}\mbox{erf}^{-1}(\kappa)\Theta_2^i}$.
\end{rem}
\begin{rem} \label{rem_falserate}
   The magnitude of the residual $r_k^i$ is tightly related to the magnitude of the biasing attack $\tau_k^i$. In other words, greater measurement bias $\tau_k^i$ results in greater residual $r_k^i$ exceeding the threshold $\theta_\kappa$ with higher attack probability $\kappa$ and lower probability of false alarm  ${\varkappa = 1-\kappa}$. 
\end{rem}

Recall from Remark~\ref{rem_noise} that, unlike \cite{kim2018detection,pajic2015attack,chong2015observability,lee2015secure,shoukry2017secure} considering a fixed (deterministic) threshold based on the upper bound on $\pmb{\zeta}_k$, Eq. \eqref{eq_thresold} assigns probability $\kappa$ to the threshold $\theta_\kappa$ with no such upper bound assumption on the noise terms, implying the probabilistic threshold design.

\subsection{Attack Detection and Mitigation Logic} \label{sec_main_alg}
Recall that, following  Lemma~\ref{lem_kron}, the connectivity of the $\alpha$, $\beta$, and $\gamma$-agents over $\mc{G}_\alpha$ and $\mc{G}_\beta$ results in the next lemma.
\begin{lem} \label{lem_isolate}
    Following the connectivity condition in Lemma~\ref{lem_kron} and residual formulations in \eqref{eq_r}-\eqref{eq_r_eta},
    \begin{enumerate} [(i)]
        \item In case of having no $\alpha$-agent\footnote{Number of $\alpha$-agents is equal to the rank-deficiency of the system matrix $A$ \cite{icassp13}. Therefore, for a full-rank system the associated distributed estimator has no $\alpha$-agent \cite{TCNS2020}.  },  attack at any $\beta$ or $\gamma$-agent is isolated.
        \item For isolation of attack in presence of an $\alpha$-agent $j$, the gain matrix $K$ needs to satisfy,
\begin{align} \label{eq_Kalpha}
\left |\frac{\mb{c}^\top_i K_i \mb{c}_j}{\mb{c}^\top_j K_j \mb{c}_j-1} \right | \leq \epsilon, ~ \mbox{for}~ i \neq j ,
\end{align}
where $0\leq {\epsilon < 1}$ is a pre-specified constant determining the residual ratio.
    \end{enumerate}
\end{lem}
\begin{proof}
From Lemma~\ref{lem_kron}, in absence of any $\alpha$-agent, ${\mathcal{N}_\alpha(i)=\{i\}}$ for any agent $i$ of type $\beta$ and $\gamma$. Thus, from  \eqref{eq_r}-\eqref{eq_r_eta}, 
biasing attack ${\tau^i_k \neq 0}$  at a $\beta$ or $\gamma$-agent $i$ only affects the residual $r_k^i$. This implies that $r^i_k$ is biased while $r^j_k$ (${j\neq i}$) is unbiased, implying that attack  $\tau^i_k$ is isolated at any $\beta$/$\gamma$-agent. On the other hand, in the presence of an $\alpha$-agent $j$ subject to attack $\tau_k^j\neq 0$, Eq. \eqref{eq_r}-\eqref{eq_r_eta} implies that the residual $r_k^i$ at every agent $i$ is affected  by the attack at agent ${j \in \mathcal{N}_\alpha(i)}$ via the term $\mb{c}^\top_i K_i \mb{c}_j$, while  the residual $r_k^j$ at $\alpha$-agent $j$ 
is affected by the factor ${{\mb{c}^\top_j K_j \mb{c}_j-1}}$. 
Therefore, Eq. \eqref{eq_Kalpha} ensures that ${|\frac{r_k^j}{r_k^i}|>\frac{1}{\epsilon} > 1}$ (for ${i \neq j}$), implying greater residual at  $\alpha$-agent $j$ by factor $\frac{1}{\epsilon}$. This constraint ensures that the  attack  can  be isolated at every $\alpha$-agent $j$. 
\end{proof}
Following Lemma~\ref{lem_threshold} and~\ref{lem_isolate}, for the attacked agent $i$ (of any type) the residual $r_k^i$ is (more) biased over  $\theta_\kappa$ in \eqref{eq_thresold}, while the residuals at other agents are less biased (or unbiased). Largest $\kappa$ such that ${r_k^i\geq \theta_\kappa}$ declares the probability of attack (or probability of false alarm  ${1-\kappa}$). Likewise, from Remark~\ref{rem_erf} and~\ref{rem_falserate}, the attack detection logic can be designed for a \textit{given} false alarm rate ${\varkappa_i}$ (and probabilistic  threshold $\theta_{\kappa_i}$) at  sensor $i$. Then, similar to the deterministic case,
the following hypothesis testing locally declares ``Attack`` or ``No-Attack`` at sensor $i$ (under certain false alarm rate ${\varkappa_i}$),
\begin{equation}
  \text{If}~ \left\{
  \begin{array}{@{}l}
     r_{k}^i\geq \theta_{\kappa_i} \\
     r_{k}^i<\theta_{\kappa_i} 
  \end{array}\right. ~\text{Then}~\left\{
  \begin{array}{@{}l}
     \mc{H}^i_1: \text{Attack Detected} \\
     \mc{H}^i_0: \text{No Attack} 
  \end{array}\right. 
\end{equation} 
\begin{rem} \label{rem_consistency}
A relevant concept is \textit{nodal/local consistency} of measurement/prediction information (data) set at agent $i$ and $j \in \mc{N}_{\alpha}(i) \cup \mc{N}_\beta(i)$ at every time $k$, denoted by  $\mc{I}_k^i,\mc{I}_k^j$ \cite{khan2013secure}.
Recall that nodal consistency checks the \textit{statistical consistency} of  $\mc{I}_{k}^i$  with the information $\mc{I}_{[k-T,k]}^i$ over a sliding time-window $T$, declaring that $\mc{I}_{k}^i$ is trustable or not. In this direction, one can track the information  over such time-window $T$  and apply, for example, a \textit{chi-square detector} on the residuals over $T$ \cite{icas21_attack} instead of instantaneous residuals \eqref{eq_r}. Local consistency, on the other hand, checks the \textit{statistical consistency} of the common information (e.g., on the shared observable subspace) between $\mc{I}_{[k-T,k]}^i$ and received information $\mc{I}_{[k-T,k]}^j$, $j \in \mc{N}_{\alpha}(i) \cup \mc{N}_\beta(i)$, and declares if $\mc{I}_{k}^j$ is trustable or not. Note that for (necessary) $\alpha$/$\beta$-agents,  weak local consistencies imply certain loss of  observability information and degradation of estimation performance. 
\end{rem}
\begin{rem} \label{rem_mitigate}
    (\textbf{Attack mitigation}) From Section~\ref{sec_necessary}, $\alpha$/$\beta$-agents  are necessary for observability; therefore, in case of attacks, their erroneous information of their observable subsystems makes those subsystems unobservable to all agents, causing unstable estimation error. To recover the loss of observability, recall that the states in the same parent SCC $\mc{S}^p_l$ and in the same contraction $\mc{C}_l$ are \textit{observationally-equivalent}, in the sense that measurement of two states in $\mc{S}^p_l$ or in $\mc{C}_l$  provide information on the same observable subsystem. In other words, the information $\mc{I}_k^i,\mc{I}_k^j$ offered by two state measurements (agents~$i,j$) are said to be observationally-equivalent if they equally contribute to the rank recovery of the \textit{observability Gramian} (see detailed definition in \cite{icassp2016,tnse18}). In this regard, for attack mitigation, the biased measurement can be replaced with a new measurement of an observationally-equivalent state in $\mc{S}^p_l$ or $\mc{C}_l$. Note that, after mitigating the attacks, the performance analysis follows as in Section~\ref{sec_perf}.
\end{rem}
\begin{rem} \label{rem_cost}
    (\textbf{Cost-optimal mitigation}) 
    Given an observationally-equivalent set of state nodes $\mc{S}^p_l$  or $\mc{C}_l$,  the substitute/replacement state measurement can be chosen based on its sensing cost. Combinatorial optimization strategies \cite{TNSE19}, e.g., the well-known \textit{Hungarian} algorithm, can be adopted to find the minimal-cost equivalent measurement to reduce the overall sensing cost. Similar arguments hold for  cost-optimal design of the multi-agent network  $\mc{G}_N= \mc{G}_\alpha \cup \mc{G}_\beta$, e.g., using the  so-called \textit{minimum spanning strong sub-graph} algorithm \cite{spl18}. 
\end{rem}
Remark~\ref{rem_mitigate} along with  Lemma~\ref{lem_threshold} and~\ref{lem_isolate} result in 
Algorithm~\ref{alg_1}.
\begin{algorithm} [t] \label{alg_1}
	\textbf{Input:} System digraph $\mc{G}_A$ and its contractions $\mc{C}$ and parent SCCs $\mc{S}^p$, $\alpha/\beta/\gamma$ classification, local estimate ${\widehat{\mb{x}}^i_{k|k}}$ at every agent ${i\in \{1,...,N\}}$.
	
    \textbf{Initialization:} $\widehat{\mb{x}}^i_{0|0}$ is set randomly at all agents $i$
    \\
    Every agent $i$ does the following:
    
    Finds the thresholds $\theta_\kappa$ based on Eq.  \eqref{eq_thresold}\;
    Finds $\widehat{\mb{x}}^i_{k|k-1}$ and $\widehat{\mb{x}}^i_{k|k}$ for $k\geq 1$ via Eq. \eqref{eq_p}-\eqref{eq_m}\;
	Finds the residual $r_k^i$ via \eqref{eq_r}\;
	\If{$r_{k>k_a}^i> \theta_\kappa$}{
		\textbf{Declares:} attack detected with probability $\kappa$\;
		\If{agent $i$ is type $\alpha$}{
			Substitutes new agent $i'$ with output from another state in the same contraction~$\mc{C}_l$\;}
		\Else{\If{agent $i$ is type $\beta$}{
				Substitutes  new agent $i'$ with output from another state in the same parent SCC~$\mc{S}^p_l$ \;}
			\Else{Remove $\gamma$-agent $i$ with no substitution\;}	
		}
	}
	\textbf{Output} Attack probability $\kappa$ and substitute agent $i'$\;	
	\caption{Attack Detection and Mitigation}
\end{algorithm}
Note that the terms $D_C$ in  \eqref{eq_err1} and $\overline{R}$ in \eqref{eq_theta1} are defined locally, i.e., 
the $i$-th diagonal block of $D_C$ and $\overline{R}$ related to agent $i$ are defined based on  received measurement information $\mb{c}_j$ and $R_j$ from  its direct neighbors (summation is over ${j\in \mathcal{N}_\alpha(i)}$). 
Therefore, the calculations of these terms are distributed and localized over the network.
The thresholds $\theta_\kappa$ in \eqref{eq_thresold}, agent types, and the sets of observationally-equivalent states in the system digraph $\mc{G}_A$ are determined by a central entity once off-line, then, broadcasted and transmitted to every agent. This procedure is done once 
and the information is stored at all agents; then, the agents can perform estimation and detect the attack locally with no further role of the centralized entity. 
See similar assumptions in	\cite{usman_cdc:11,khan2014collaborative} for distributed estimation/filtering.
\begin{rem} \label{rem_order}
	The  DM (Dulmage-Mendelsohn) decomposition  and DFS  (depth-first-search) or Kosaraju-Sharir algorithms  can be used, respectively,  to find contractions and SCCs (along with their topological order)
	with computational complexity $\mc{O}(n^{2.5})$ and $\mc{O}(n^2)$ \cite{murota}. 
	The residual calculation at agents is of $\mc{O}(n)$ complexity, while the complexity of the threshold design based on $2$-norm calculation is $\mc{O}(N^3n^3)$. Overall, the complexity of Algorithm~\ref{alg_1} is $\mc{O}(N^3n^3)$. 
	This polynomial order complexity suits large-scale applications.    
\end{rem}
\section{Simulation}\label{sec_sim}
For simulation we consider a dynamical system with $10$ states associated with the  system digraph $\mc{G}_A$  in Fig. \ref{fig_sim1}-(Left). The link weights in $\mc{G}_A$   are considered randomly (such that $\rho(A)>1$).
\begin{figure}[t]
	\centering
	\includegraphics[width=1.7in]{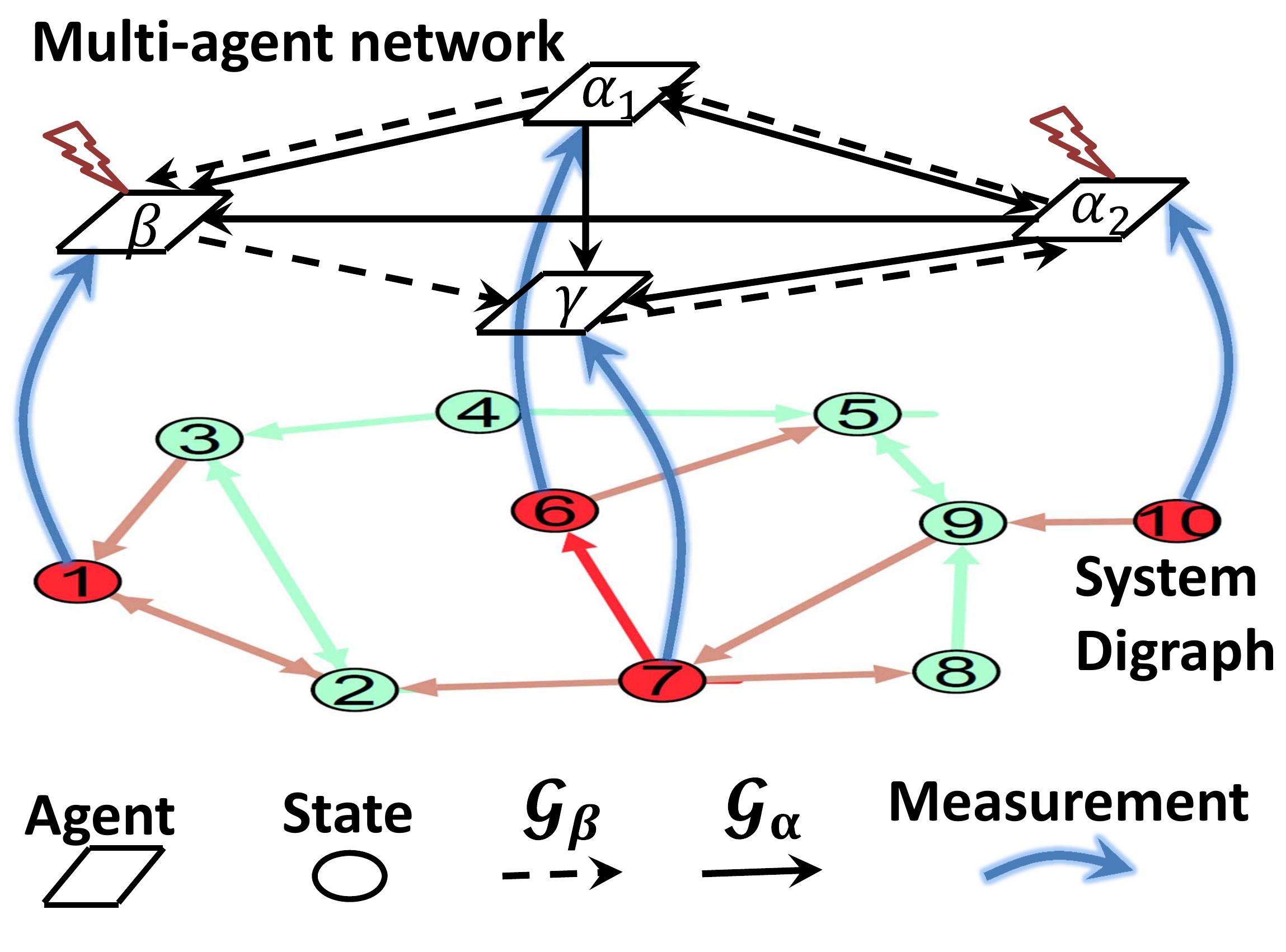}
	\includegraphics[width=1.6in] {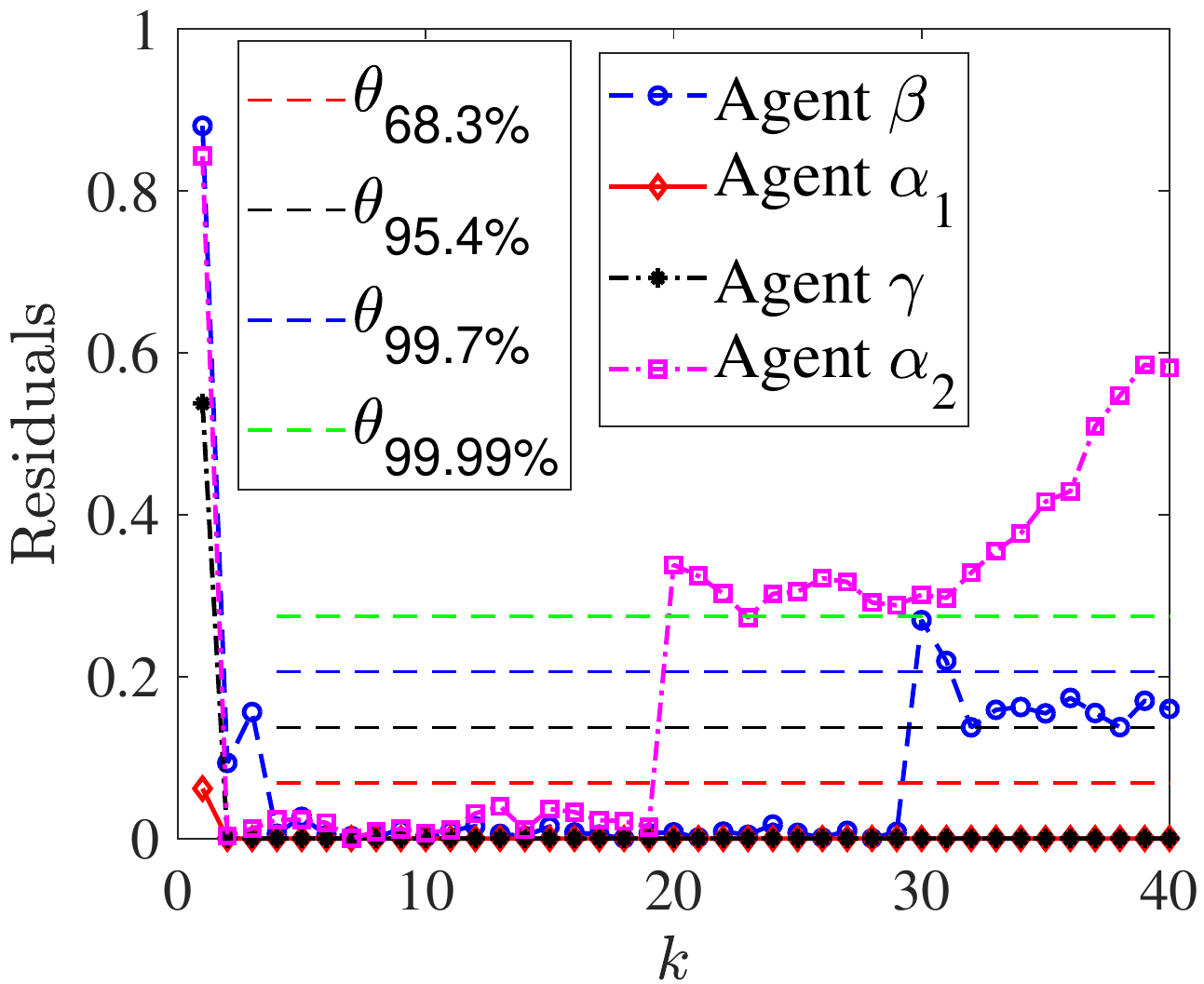}
	\caption{(Left) The multi-agent network (top network including $\mc{G}_\alpha$ and $\mc{G}_\beta$)  estimates the states of the dynamical system (bottom network) by taking output measurements of the  states  in red color.  The agents $\alpha_2$ and $\beta$ are under attack. (Right) The residuals at $4$ agents over time along with the thresholds $\theta_\kappa$ are shown. As expected, the residuals at the attacked agents are biased over the thresholds. } \vspace{-0.4pt}
	\label{fig_sim1}
\end{figure}
Following Remark~\ref{rem_order}, the contractions and parent SCCs in $\mc{G}_A$ are: ${\mc{S}^p_1=\{1,2,3\}}$, ${\mc{C}_1 = \{5,6,8\}}$, and ${\mc{C}_2= \{8,10\}}$. From Section~\ref{sec_necessary}, one output from each of these node sets ensure observability of $\mc{G}_A$. As shown in Fig.\ref{fig_sim1}-(Left), agents $\beta$, $\alpha_1$, and $\alpha_2$  take output of state $1$, $6$, and $10$, respectively, along with a redundant agent $\gamma$ with output of state $7$ (which is not necessary for observability). Following Section~\ref{sec_est}, the  network $\mc{G}_\beta$ is considered as a cycle, 
while in $\mc{G}_\alpha$ agents $\alpha_2$ and $\alpha_1$ are two hubs of the network. Each agent adopts the proposed protocol \eqref{eq_p}-\eqref{eq_m} to estimate all $10$ system states (with partial observability via its measurement and neighboring information). The link weights in ${\mc{G}_\beta}$ (the nonzero $W_{ij}$s)  are chosen randomly such that $W$ is row-stochastic. The noise terms follow ${\nu = \mc{N}(0,0.01 \mb{1}_{nn})}$ and ${\zeta = \mc{N}(0,0.01 I_N)}$. The block-diagonal gain $K$ is determined via heuristic LMIs such that, for example: ${|\mb{c}^\top_{\beta} K_{\beta} \mb{c}_{\alpha_1}| = 0.008}$,  ${|\mb{c}^\top_{\gamma} K_{\gamma} \mb{c}_{\alpha_1}|= 0.00001}$,   ${|\mb{c}^\top_{\alpha_2} K_{\alpha_2} \mb{c}_{\alpha_1}|=0.005}$,   ${|\mb{c}^\top_{\alpha_1} K_{\alpha_1} \mb{c}_{\alpha_1}|=0.24}$, satisfying Lemma~\ref{lem_isolate} for any  ${0.011 \leq \epsilon < 1 }$ with $j$ as agent $\alpha_1$ in \eqref{eq_Kalpha}. Likewise, ${0.01 \leq \epsilon < 1 }$ for agent $\alpha_2$, implying that, for this given $K$, the attack-related portion of the residual at attacked agent $\alpha_2$ is almost $100$ times greater than the residuals at other (non-attacked) agents. Therefore, any attack at agents $\alpha_2,\alpha_1$ can be isolated. 
The parameters in Eq. \eqref{eq_theta1} are  ${a_1 = 2.937}$,  ${a_2 = 0.183}$,  ${b=0.682}$, 
which result in  ${\Theta_1 = 0.068}$ and  ${\Theta_2 = 0.078}$.
We consider fixed attack   ${\tau_{k\geq 30}=1}$ at agent  $\beta$ (following Assumption~(iv)) along with an auto-regressive non-stationary attack for $k\geq 20$  at agent $\alpha_2$ in the form ${\tau_{k+2}=2\tau_{k+1}-\tau_{k}+\vartheta}$ with ${\tau_{20}=\tau_{21}=0.3}$ and ${\vartheta \in [0,0.02]}$ as a uniform random variable. 
The residuals~\eqref{eq_r} (shown in Fig.~\ref{fig_sim1}-(Right)) at the attacked agents $\beta$ and $\alpha_2$ are biased, respectively, over ${\theta_{95.4\%}}$ and $\theta_{99.99\%}$, implying false alarm probabilities\footnote{The auto-regressive attack is given as an example of possible extension of the results to the case of non-stationary attacks, where the attack probabilities can be approximated by Lemma~\ref{lem_threshold}. } approximately less than $4.6\%$ and $0.01\%$. 

\textbf{Comparison with recent literature:} next, we use the estimation and detection strategy in \cite{he2019secure,he2020secure} for comparison. Recall that from Remark~\ref{rem_scale}, the distributed observer in  \cite{he2019secure,he2020secure} is a double time-scale protocol, which requires many iterations of consensus between every two time-steps of system dynamics. Therefore, it needs much faster information sharing/processing rate as compared to the proposed protocol \eqref{eq_p}-\eqref{eq_m}. \textit{The reason for choosing  \cite{he2019secure,he2020secure} for comparison study is that  double time-scale protocols make similar relaxed observability assumption as Assumption~(ii) in Section~\ref{sec_ass} (irrespective of system rank-deficiency). This is in contrast to many exisitng single time-scale protocols, e.g.,  \cite{usman_tsp:07,cattivelli2008diffusion,chen2016dynamic,chen2018resilient,battistelli_cdc,sayedtu12,nuno-suff.ness}, which assume that the underlying system is observable in the neighborhood of each agent and/or is full-rank. In other words,  the mentioned references generally  require more network connectivity,  and therefore, do not result in steady-state stable error over the given $\mc{G}_\alpha$ and $\mc{G}_\beta$ networks in Fig.~\ref{fig_sim1}-(Left).} We set the  parameters in \cite{he2019secure,he2020secure}
as in Table~\ref{tab_par} (which seem to provide the best outcome).
	\begin{table} [t] 
		\centering
		\caption{Parameter values for the detection and estimation protocol in \cite{he2019secure,he2020secure}.}
		\label{tab_par}
		\begin{tabular}{|c|c||c|c|c|c|} 
			\hline
			$L$& $40$ & $\alpha$ & $2$ &$\beta$ & $0.4$ \\
			\hline
			$\|A\|$ &  $1.35$ & $\gamma$ & $0.57$  & $N$& $4$ \\
			\hline
			$s$ & $2$  &$b_w$ & $0.06$ &$b_v$ & $0.06$ \\
			\hline
			$\lambda_0$ & $1.95$  &$\eta_0$ & $0.1$ &$\rho_{t_0}$ & $0.1$ \\		
			\hline
			\hline
		\end{tabular}
	\end{table}
In this simulation, agents need to perform ${L=40}$ consensus iterations for estimation/detection, which requires $40$-times faster communication and computation rate as compared to the proposed protocol \eqref{eq_p}-\eqref{eq_m}. The results are shown in Fig. \ref{fig_sim2}-(Left). 
\begin{figure}[t]
	\centering
	\includegraphics[width=1.6in]{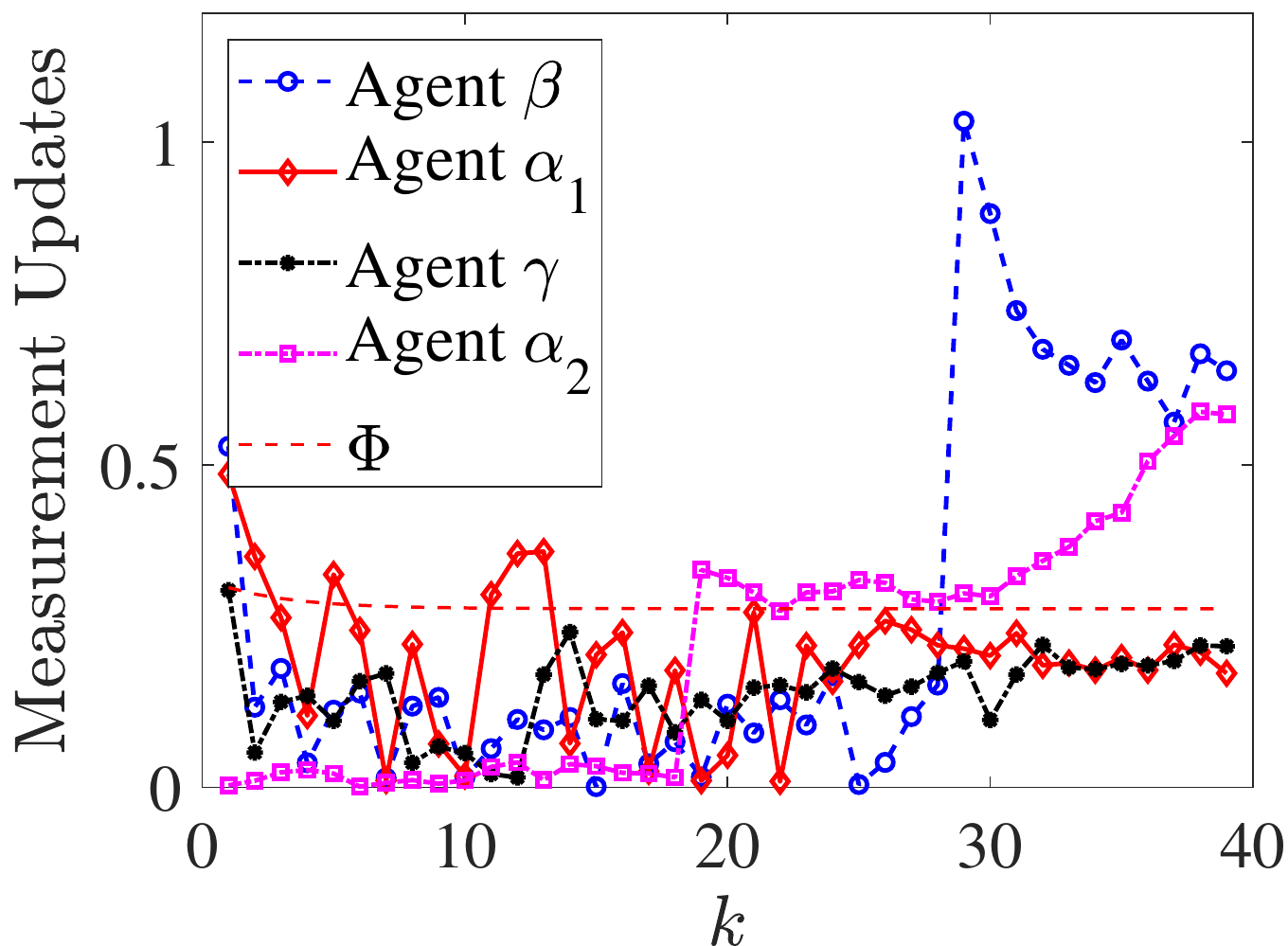}
	\includegraphics[width=1.7in]{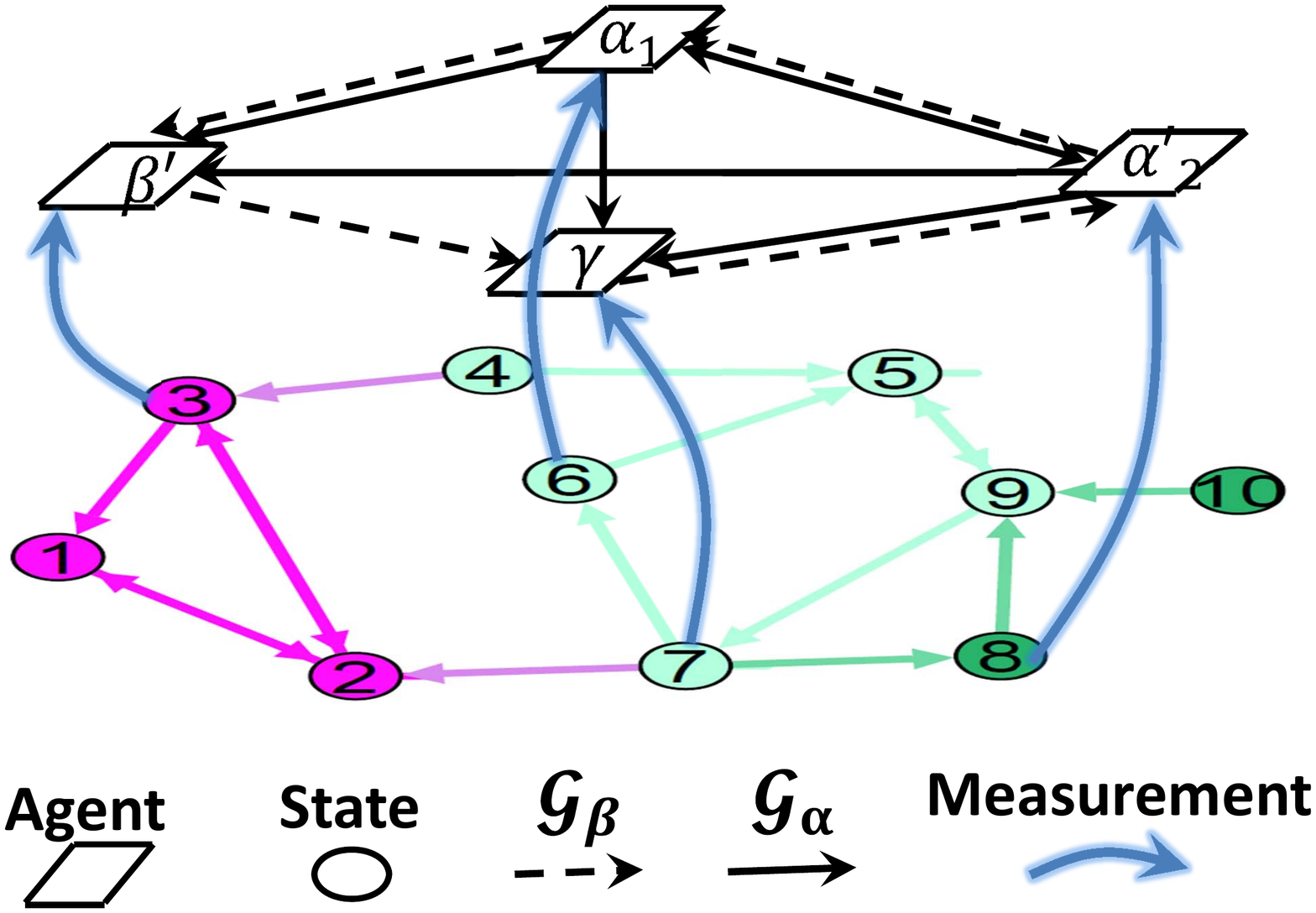}
	\caption{ (Left) This figure shows the measurement-updates at all agents based on the methodology in \cite{he2019secure,he2020secure}. The attack is detected via the threshold $\Phi$. Clearly, the protocol in \cite{he2019secure,he2020secure} with parameters given in Table~\ref{tab_par} detects both attacks at agents $\alpha_2$ and $\beta$, while also raising false alarm on agent $\alpha_1$ at some times. In contrast, our proposed detection strategy only raise alarm on the attacked agents as shown in Fig.~\ref{fig_sim1}-(Right). (Right) Using Algorithm~\ref{alg_1}, the detected attacks are mitigated by adding  equivalent agents $\alpha_2^{\prime}$ and $\beta^{\prime}$ to  recover the loss of observability. The new agents $\alpha_2^\prime$ and $\beta^\prime$ measure observationally-equivalent states, respectively, in the same contraction (green  nodes) and in the same parent SCC (purple  nodes). } 
	\label{fig_sim2} 
\end{figure}
Following the attack detection logic in \cite{he2019secure,he2020secure}, the agents can detect possible attacks if their \textit{measurement-updates} are over a certain threshold $\Phi$. From  Fig.~\ref{fig_sim2}-(Left), both attacks are detected, while also falsely alarming  attack at agent $\alpha_1$ at some times.

\textbf{Attack mitigation and performance analysis:} next, using the mitigation strategy in Algorithm~\ref{alg_1}, we replace the detected attacked agents $\beta$ and $\alpha_2$ with substitute agents $\beta^\prime$ and $\alpha_2^\prime$, respectively measuring observationally-equivalent state $3$ in $\mc{S}^p_1$ and state $8$ in $\mc{C}_2$. The connectivity of the new agents follows the same connectivity of $\mc{G}_\alpha$ and $\mc{G}_\beta$ as shown in Fig.~\ref{fig_sim2}-(Right).
\begin{figure}[t]
	\centering
	\includegraphics[width=1.7in]{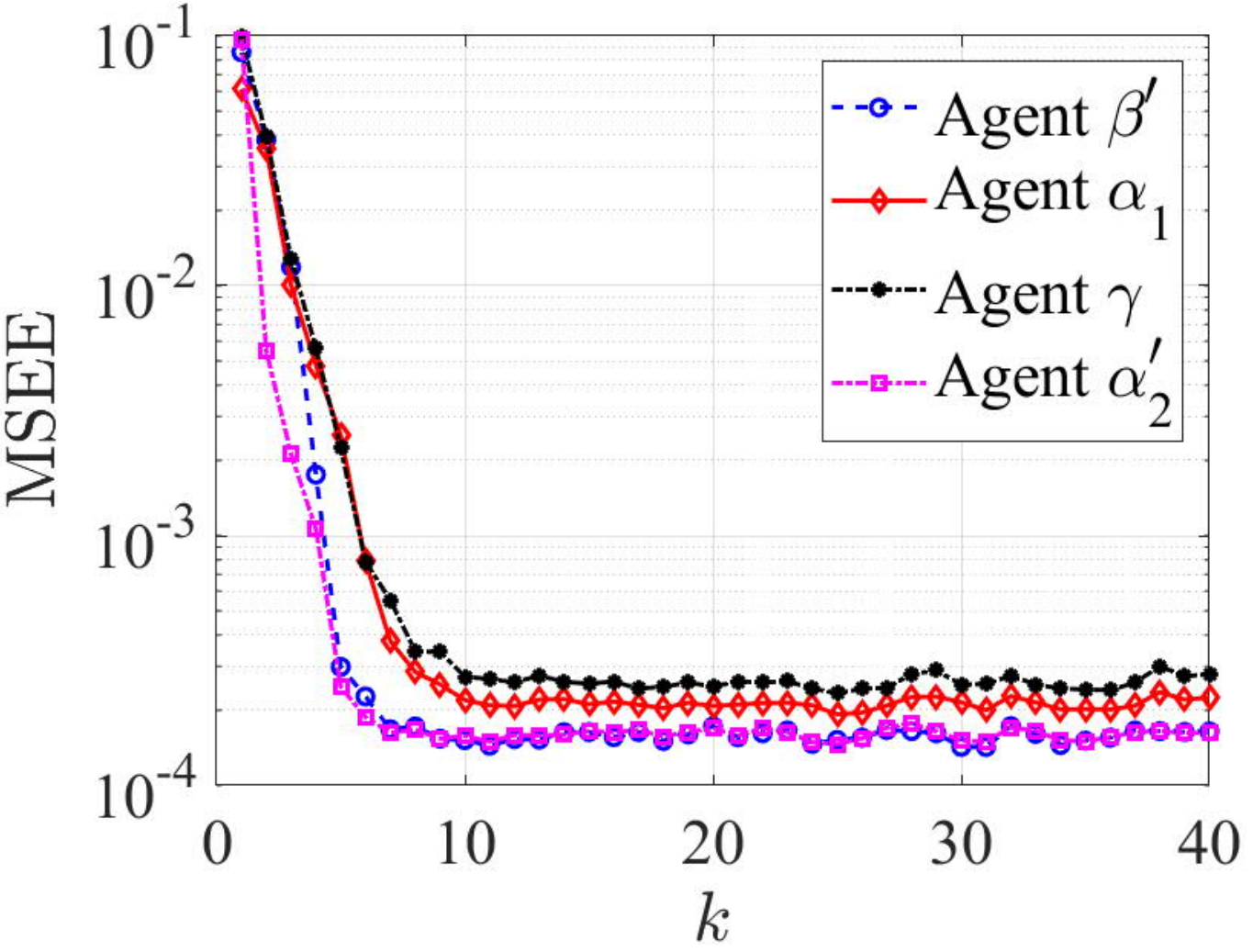}
	\includegraphics[width=1.7in] {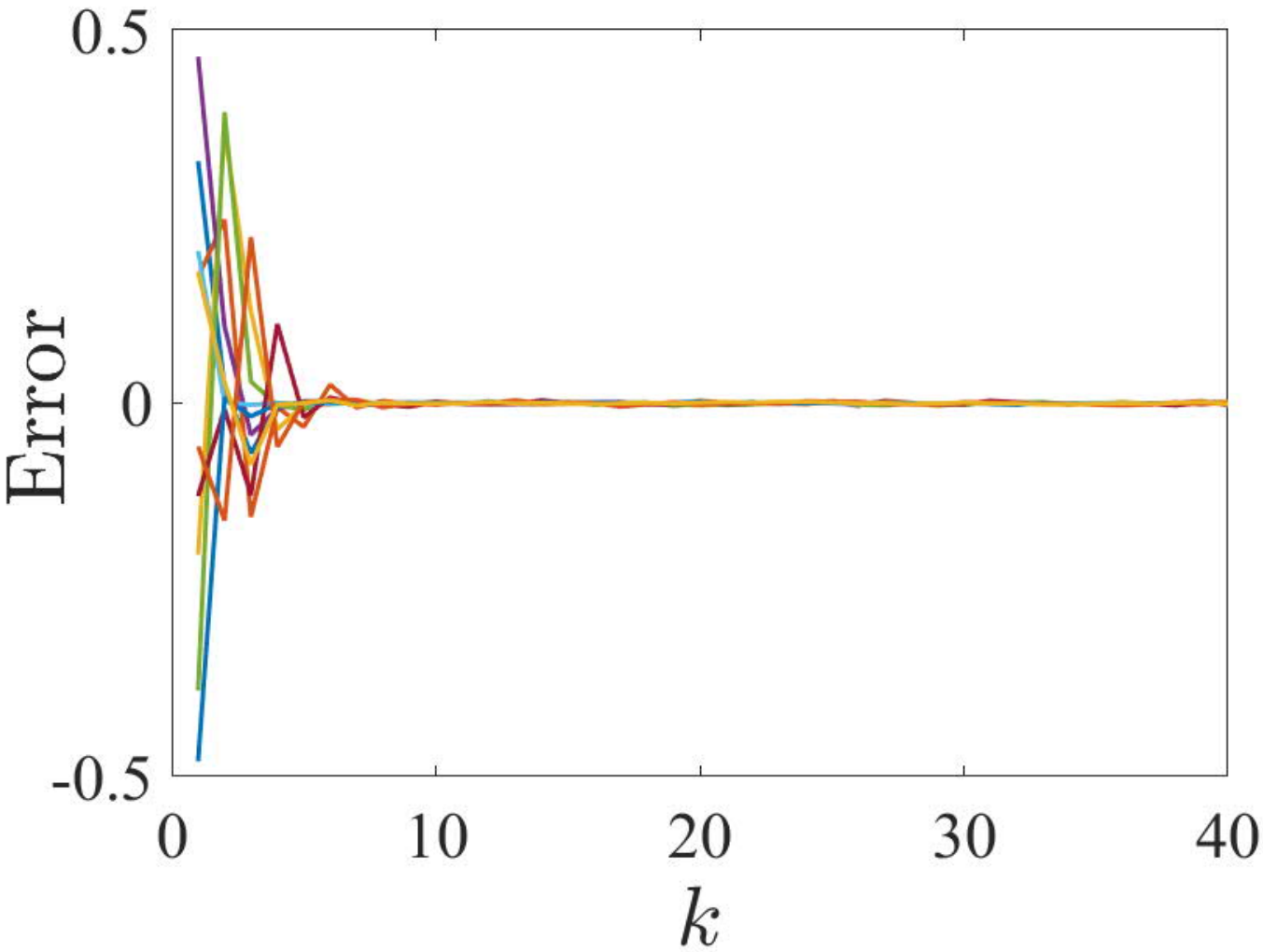}
	\caption{ (Left) This figure presents Monte-Carlo time-evolution of the MSEEs (in log-scale) at $4$ agents after attack mitigation. The bounded steady-state MSEEs imply observable estimation/filtering. (Right) This figure shows the  (Monte-Carlo) time-evolution of the estimation errors of all $10$ states at agent $\beta^\prime$, which are unbiased in steady-state.  }
	\label{fig_perform}
\end{figure}
We perform Monte-Carlo simulation (averaged over $100$ repetitions) of the proposed protocol \eqref{eq_p}-\eqref{eq_m} for the attack-mitigated case of Fig.~\ref{fig_sim2}-(Right).
The mean-square performance and mean performance are shown in Fig.~\ref{fig_perform}. As it is clear, the mean-square estimation errors (MSEEs) are bounded steady-state stable at all agents as expected from Lemma~\ref{lem_Ee2}. Further, from Lemma~\ref{lem_Ee}, the steady-state errors at all agents are unbiased; Fig.~\ref{fig_perform}-(Right) shows unbiased state errors at agent $\beta^\prime$ as an example.

\section{Conclusion}\label{sec_con}
This paper considers a decentralized attack detection  over distributed estimation networks. 
The detection, isolation, and mitigation strategy is designed specifically for $\alpha$, $\beta$, and $\gamma$-agents in polynomial-order complexity.
As future research direction, network reconfiguration \cite{csl2020,icas21_contraction}  to reduce attack vulnerability and design of attack-tolerant/resilient engineered networks is promising. Further, one can track the history of  residuals (for general rank-deficient systems) over a sliding time-window (known as \textit{stateful detection} \cite{giraldo2018survey}), similar to $\chi^2$-detection  \cite{icas21_attack} or trust-index evolution  \cite{khan2013secure}.

\bibliographystyle{IEEEbib}
\bibliography{bibliography}

\begin{thebibliography}{10}

\bibitem{asefi2021application}
S.~Asefi, Y.~Madhwal, Y.~Yanovich, and E.~Gryazina,
\newblock ``Application of blockchain for secure data transmission in
  distributed state estimation,''
\newblock {\em arXiv preprint arXiv:2104.04232}, 2021.

\bibitem{camsap11}
U.~A. Khan and M.~Doostmohammadian,
\newblock ``A sensor placement and network design paradigm for future smart
  grids,''
\newblock in {\em 4th International Workshop on Computational Advances in
  Multi-Sensor Adaptive Processing}, San Juan, Puerto Rico, Dec. 2011, pp.
  137--140.

\bibitem{yang2020distributed}
W.~Yang, W.~Luo, and X.~Zhang,
\newblock ``Distributed secure state estimation under stochastic linear
  attacks,''
\newblock {\em IEEE Transactions on Network Science and Engineering}, 2020.

\bibitem{acc13}
M.~Doostmohammadian and U.~A. Khan,
\newblock ``Topology design in networked estimation: a generic approach,''
\newblock in {\em American Control Conference}, Washington, DC, Jun. 2013, pp.
  4140--4145.

\bibitem{isj2020}
M.~Doostmohammadian, H.~R. Rabiee, and U.~A. Khan,
\newblock ``Cyber-social systems: modeling, inference, and optimal design,''
\newblock {\em IEEE Systems Journal}, vol. 14, no. 1, pp. 73--83, 2020.

\bibitem{xu2016distributed}
S.~Xu, R.~C. De~Lamare, and H.~V. Poor,
\newblock ``Distributed estimation over sensor networks based on distributed
  conjugate gradient strategies,''
\newblock {\em IET Signal Processing}, vol. 10, no. 3, pp. 291--301, 2016.

\bibitem{pasqualetti2013attack}
F.~Pasqualetti, F.~D{\"o}rfler, and F.~Bullo,
\newblock ``Attack detection and identification in cyber-physical systems,''
\newblock {\em IEEE transactions on automatic control}, vol. 58, no. 11, pp.
  2715--2729, 2013.

\bibitem{shiri2018distributed}
H.~Shiri, M.~A. Tinati, M.~Codreanu, and G.~Azarnia,
\newblock ``Distributed sparse diffusion estimation with reduced communication
  cost,''
\newblock {\em IET Signal Processing}, vol. 12, no. 8, pp. 1043--1052, 2018.

\bibitem{asilomar14}
M.~Doostmohammadian and U.~A. Khan,
\newblock ``Vulnerability of {CPS} inference to {DoS} attacks,''
\newblock in {\em 48th IEEE Asilomar Conference on Signals, Systems, and
  Computers}, 2014, pp. 2015--2018.

\bibitem{chen2018internet}
Y.~Chen, S.~Kar, and J.~M.~F. Moura,
\newblock ``The internet of things: Secure distributed inference,''
\newblock {\em IEEE Signal Processing Magazine}, vol. 35, no. 5, pp. 64--75,
  2018.

\bibitem{csl2020}
M.~Doostmohammadian and H.~R. Rabiee,
\newblock ``On the observability and controllability of large-scale {IoT}
  networks: Reducing number of unmatched nodes via link addition,''
\newblock {\em IEEE Control Systems Letters}, vol. 5, no. 5, pp. 1747--1752,
  2020.

\bibitem{pandey2020fault}
O.~J. Pandey, V.~Gautam, H.~H. Nguyen, M.~K. Shukla, and R.~M. Hegde,
\newblock ``Fault-resilient distributed detection and estimation over a sw-wsn
  using lcmv beamforming,''
\newblock {\em IEEE Transactions on Network and Service Management}, vol. 17,
  no. 3, pp. 1758--1773, 2020.

\bibitem{guo2020unsupervised}
Y.~Guo, T.~Ji, Q.~Wang, L.~Yu, G.~Min, and P.~Li,
\newblock ``Unsupervised anomaly detection in iot systems for smart cities,''
\newblock {\em IEEE Trans. on Network Science and Engineering}, vol. 7, no. 4,
  pp. 2231--2242, 2020.

\bibitem{pequito_gsip}
S.~Pequito, S.~Kar, and A.~P. Aguiar,
\newblock ``Minimum number of information gatherers to ensure full
  observability of a dynamic social network: a structural systems approach,''
\newblock in {\em IEEE Global Conference on Signal and Information Processing},
  2014, pp. 750--753.

\bibitem{SNAM20}
M.~Doostmohammadian, H.~R. Rabiee, and U.~A. Khan,
\newblock ``Centrality-based epidemic control in complex social networks,''
\newblock {\em Social Network Analysis and Mining}, vol. 10, pp. 1--11, 2020.

\bibitem{icas21_attack}
M.~Doostmohammadian, T.~Charalambous, M.~Shafie-khah, N.~Meskin, and U.~A.
  Khan,
\newblock ``Simultaneous distributed estimation and attack detection/isolation
  in social networks: Structural observability, {K}ronecker-product network,
  and chi-square detector,''
\newblock in {\em 1st IEEE International Conference on Autonomous Systems
  (ICAS)}, 2021,
\newblock (accepted) arXiv preprint arXiv:2105.10639.

\bibitem{dehghani2020deep}
M.~Dehghani, A.~Kavousi-Fard, M.~Dabbaghjamanesh, and O.~Avatefipour,
\newblock ``Deep learning based method for false data injection attack
  detection in ac smart islands,''
\newblock {\em IET Generation, Transmission \& Distribution}, vol. 14, no. 24,
  pp. 5756--5765, 2020.

\bibitem{cui2012coordinated}
S.~Cui, Z.~Han, S.~Kar, T.~T. Kim, H~V. Poor, and A.~Tajer,
\newblock ``Coordinated data-injection attack and detection in the smart grid:
  A detailed look at enriching detection solutions,''
\newblock {\em IEEE Signal Processing Magazine}, vol. 29, no. 5, pp. 106--115,
  2012.

\bibitem{rawat2015detection}
D.~B. Rawat and C.~Bajracharya,
\newblock ``Detection of false data injection attacks in smart grid
  communication systems,''
\newblock {\em IEEE Signal Processing Letters}, vol. 22, no. 10, pp.
  1652--1656, 2015.

\bibitem{drayer2019detection}
E.~Drayer and T.~Routtenberg,
\newblock ``Detection of false data injection attacks in smart grids based on
  graph signal processing,''
\newblock {\em IEEE Systems Journal}, vol. 14, no. 2, pp. 1886--1896, 2020.

\bibitem{chakravorti2017detection}
T.~Chakravorti, R.~K. Patnaik, and P.~K. Dash,
\newblock ``Detection and classification of islanding and power quality
  disturbances in microgrid using hybrid signal processing and data mining
  techniques,''
\newblock {\em IET Signal Processing}, vol. 12, no. 1, pp. 82--94, 2017.

\bibitem{luo2019detection}
X.~Luo, X.~Wang, X.and~Pan, and X.~Guan,
\newblock ``Detection and isolation of false data injection attack for smart
  grids via unknown input observers,''
\newblock {\em IET Generation, Transmission \& Distribution}, vol. 13, no. 8,
  pp. 1277--1286, 2019.

\bibitem{babu2016optimal}
R.~Babu and B.~Bhattacharyya,
\newblock ``Optimal allocation of phasor measurement unit for full
  observability of the connected power network,''
\newblock {\em International Journal of Electrical Power \& Energy Systems},
  vol. 79, pp. 89--97, 2016.

\bibitem{liu2011false}
Y.~Liu, P.~Ning, and M.~K. Reiter,
\newblock ``False data injection attacks against state estimation in electric
  power grids,''
\newblock {\em ACM Transactions on Information and System Security}, vol. 14,
  no. 1, pp. 1--33, 2011.

\bibitem{6712169}
J.~Chen, W.~Li, C.~Wen, J.~Teng, and P.~Ting,
\newblock ``Efficient identification method for power line outages in the smart
  power grid,''
\newblock {\em IEEE Transactions on Power Systems}, vol. 29, no. 4, pp.
  1788--1800, 2014.

\bibitem{usman_tsp:07}
U.~A. Khan and J.~M.~F. Moura,
\newblock ``Distributing the {K}alman filter for large-scale systems,''
\newblock {\em IEEE Transactions on Signal Processing}, vol. 56, no. 10, pp.
  4919--4935, Oct. 2008.

\bibitem{cattivelli2008diffusion}
F.~S. Cattivelli, C.~G. Lopes, and A.~H. Sayed,
\newblock ``Diffusion strategies for distributed kalman filtering: formulation
  and performance analysis,''
\newblock {\em Proc. Cognitive Information Processing}, pp. 36--41, 2008.

\bibitem{flock}
R.~Olfati-Saber and P.~Jalalkamali,
\newblock ``Collaborative target tracking using distributed kalman filtering on
  mobile sensor networks,''
\newblock in {\em American Control Conference}, San Francisco, CA, Jun. 2011.

\bibitem{deghat2019detection}
M.~Deghat, V.~Ugrinovskii, I.~Shames, and C.~Langbort,
\newblock ``Detection and mitigation of biasing attacks on distributed
  estimation networks,''
\newblock {\em Automatica}, vol. 99, pp. 369--381, 2019.

\bibitem{milovsevivc2017analysis}
J.~Milo{\v{s}}evi{\v{c}}, T.~Tanaka, H.~Sandberg, and K.~H. Johansson,
\newblock ``Analysis and mitigation of bias injection attacks against a kalman
  filter,''
\newblock {\em IFAC-Papers OnLine}, vol. 50, no. 1, pp. 8393--8398, 2017.

\bibitem{chen2016dynamic}
Y.~Chen, S.~Kar, and J.~M.~F. Moura,
\newblock ``Dynamic attack detection in cyber-physical systems with side
  initial state information,''
\newblock {\em IEEE Trans. on Automatic Control}, vol. 62, no. 9, pp.
  4618--4624, 2016.

\bibitem{chen2018resilient}
Y.~Chen, S.~Kar, and J.~M.~F. Moura,
\newblock ``Resilient distributed estimation: Sensor attacks,''
\newblock {\em IEEE Transactions on Automatic Control}, vol. 64, no. 9, pp.
  3772--3779, 2018.

\bibitem{battistelli_cdc}
G.~Battistelli, L.~Chisci, G.~Mugnai, A.~Farina, and A.~Graziano,
\newblock ``Consensus-based algorithms for distributed filtering,''
\newblock in {\em 51st IEEE Conference on Decision and Control}, 2012, pp.
  794--799.

\bibitem{sayedtu12}
S.~Tu and A.~Sayed,
\newblock ``Diffusion strategies outperform consensus strategies for
  distributed estimation over adaptive networks,''
\newblock {\em IEEE Trans. on Signal Proc.}, vol. 60, no. 12, pp. 6217--6234,
  2012.

\bibitem{nuno-suff.ness}
S.~Park and N.~Martins,
\newblock ``Necessary and sufficient conditions for the stabilizability of a
  class of {LTI} distributed observers,''
\newblock in {\em 51st IEEE Conference on Decision and Control}, 2012, pp.
  7431--7436.

\bibitem{giraldo2018survey}
J.~Giraldo, D.~Urbina, A.~Cardenas, J.~Valente, M.~Faisal, J.~Ruths, N.~O.
  Tippenhauer, H.~Sandberg, and R.~Candell,
\newblock ``A survey of physics-based attack detection in cyber-physical
  systems,''
\newblock {\em ACM Computing Surveys}, vol. 51, no. 4, pp. 1--36, 2018.

\bibitem{guan2017distributed}
Y.~Guan and X.~Ge,
\newblock ``Distributed attack detection and secure estimation of networked
  cyber-physical systems against false data injection attacks and jamming
  attacks,''
\newblock {\em IEEE Transactions on Signal and Information Proc. over
  Networks}, vol. 4, no. 1, pp. 48--59, 2017.

\bibitem{khan_cdc:2010}
U.~A. Khan, S.~Kar, A.~Jadbabaie, and J.~M.~F. Moura,
\newblock ``On connectivity, observability, and stability in distributed
  estimation,''
\newblock in {\em 49th IEEE conference on decision and control}, 2010, pp.
  6639--6644.

\bibitem{jstsp}
M.~Doostmohammadian and U.~Khan,
\newblock ``On the genericity properties in distributed estimation: Topology
  design and sensor placement,''
\newblock {\em IEEE Journal of Selected Topics in Signal Processing}, vol. 7,
  no. 2, pp. 195--204, 2013.

\bibitem{woude:03}
J.~M. Dion, C.~Commault, and J.~van~der Woude,
\newblock ``Generic properties and control of linear structured systems: {A}
  survey,''
\newblock {\em Automatica}, vol. 39, pp. 1125--1144, Mar. 2003.

\bibitem{jstsp14}
M.~Doostmohammadian and U.~Khan,
\newblock ``Graph-theoretic distributed inference in social networks,''
\newblock {\em IEEE Journal of Selected Topics in Signal Processing}, vol. 8,
  no. 4, pp. 613--623, Aug. 2014.

\bibitem{icassp2016}
M.~Doostmohammadian and U.~A. Khan,
\newblock ``Measurement partitioning and observational equivalence in state
  estimation,''
\newblock in {\em IEEE International Conference on Acoustics, Speech and Signal
  Processing (ICASSP)}, 2016, pp. 4855--4859.

\bibitem{pereira2013diffusion}
S.~S. Pereira, R.~L{\`o}pez-Valcarce, and A.~Pag{\`e}s-Zamora,
\newblock ``A diffusion-based {EM} algorithm for distributed estimation in
  unreliable sensor networks,''
\newblock {\em IEEE Signal Processing Letters}, vol. 20, no. 6, pp. 595--598,
  2013.

\bibitem{spl17}
M.~Doostmohammadian, H.~R. Rabiee, H.~Zarrabi, and U.~A. Khan,
\newblock ``Distributed estimation recovery under sensor failure,''
\newblock {\em IEEE Signal Processing Letters}, vol. 24, no. 10, pp.
  1532--1536, 2017.

\bibitem{dutta2019resilient}
R.~G. Dutta, T.~Zhang, and Y.~Jin,
\newblock ``Resilient distributed filter for state estimation of cyber-physical
  systems under attack,''
\newblock in {\em American Control Conference (ACC)}. IEEE, 2019, pp.
  5141--5147.

\bibitem{mitra2019resilient}
A.~Mitra, J.~Richards, S.~Bagchi, and S.~Sundaram,
\newblock ``Resilient distributed state estimation with mobile agents:
  overcoming byzantine adversaries, communication losses, and intermittent
  measurements,''
\newblock {\em Autonomous Robots}, vol. 43, no. 3, pp. 743--768, 2019.

\bibitem{mustafa2019secure}
A.~Mustafa and H.~Modares,
\newblock ``Secure event-triggered distributed kalman filters for state
  estimation,''
\newblock {\em arXiv preprint arXiv:1901.06746}, 2019.

\bibitem{wen2018distributed}
F.~Wen and Z.~Wang,
\newblock ``Distributed kalman filtering for robust state estimation over
  wireless sensor networks under malicious cyber attacks,''
\newblock {\em Digital Signal Processing}, vol. 78, pp. 92--97, 2018.

\bibitem{yang2020adversary}
Z.~Yang, A.~Gang, and W.~Bajwa,
\newblock ``Adversary-resilient distributed and decentralized statistical
  inference and machine learning: An overview of recent advances under the
  byzantine threat model,''
\newblock {\em IEEE Signal Proc. Magazine}, vol. 37, no. 3, pp. 146--159, 2020.

\bibitem{su2019finite}
L.~Su and S.~Shahrampour,
\newblock ``Finite-time guarantees for byzantine-resilient distributed state
  estimation with noisy measurements,''
\newblock {\em IEEE Transactions on Automatic Control}, 2019.

\bibitem{li2017sampled}
Q.~Li, B.~Shen, Z.~Wang, and F.~E. Alsaadi,
\newblock ``A sampled-data approach to distributed h$\infty$ resilient state
  estimation for a class of nonlinear time-delay systems over sensor
  networks,''
\newblock {\em Journal of the Franklin Institute}, vol. 354, no. 15, pp.
  7139--7157, 2017.

\bibitem{wang2014stochastically}
X.~Wang and E.~Yaz,
\newblock ``Stochastically resilient extended kalman filtering for
  discrete-time nonlinear systems with sensor failures,''
\newblock {\em International Jour. of Syst. Science}, vol. 45, no. 7, pp.
  1393--1401, 2014.

\bibitem{he2019secure}
X.~He, X.~Ren, H.~Sandberg, and K.~H. Johansson,
\newblock ``Secure distributed filtering for unstable dynamics under
  compromised observations,''
\newblock in {\em IEEE 58th Conference on Decision and Control (CDC)}. IEEE,
  2019, pp. 5344--5349.

\bibitem{he2020secure}
X.~He, X.~Ren, H.~Sandberg, and K.~H. Johansson,
\newblock ``How to secure distributed filters under sensor attacks?,''
\newblock {\em IEEE Transactions on Automatic Control}, 2021,
\newblock arXiv preprint arXiv:2004.05409.

\bibitem{kim2018detection}
J.~Kim, C.~Lee, H.~Shim, Y.~Eun, and J.~H. Seo,
\newblock ``Detection of sensor attack and resilient state estimation for
  uniformly observable nonlinear systems having redundant sensors,''
\newblock {\em IEEE Transactions on Automatic Control}, vol. 64, no. 3, pp.
  1162--1169, 2018.

\bibitem{pajic2015attack}
M.~Pajic, P.~Tabuada, I.~Lee, and G.~J. Pappas,
\newblock ``Attack-resilient state estimation in the presence of noise,''
\newblock in {\em 54th IEEE Conference on Decision and Control (CDC)}. IEEE,
  2015, pp. 5827--5832.

\bibitem{chong2015observability}
M.~S. Chong, M.~Wakaiki, and J.~P. Hespanha,
\newblock ``Observability of linear systems under adversarial attacks,''
\newblock in {\em American Control Conference (ACC)}. IEEE, 2015, pp.
  2439--2444.

\bibitem{lee2015secure}
C.~Lee, H.~Shim, and Y.~Eun,
\newblock ``Secure and robust state estimation under sensor attacks,
  measurement noises, and process disturbances: Observer-based combinatorial
  approach,''
\newblock in {\em European Control Conference (ECC)}. IEEE, 2015, pp.
  1872--1877.

\bibitem{shoukry2017secure}
Y.~Shoukry, P.~Nuzzo, A.~Puggelli, A.~L. Sangiovanni-Vincentelli, S.~A. Seshia,
  and P.~Tabuada,
\newblock ``Secure state estimation for cyber-physical systems under sensor
  attacks: A satisfiability modulo theory approach,''
\newblock {\em IEEE Transactions on Automatic Control}, vol. 62, no. 10, pp.
  4917--4932, 2017.

\bibitem{kailkhura2016data}
B.~Kailkhura, S.~Brahma, and P.~K. Varshney,
\newblock ``Data falsification attacks on consensus-based detection systems,''
\newblock {\em IEEE Transactions on Signal and Information Processing over
  Networks}, vol. 3, no. 1, pp. 145--158, 2016.

\bibitem{wang2017data}
P.~Wang, M.~Govindarasu, A.~Ashok, S.~Sridhar, and D.~McKinnon,
\newblock ``Data-driven anomaly detection for power system generation
  control,''
\newblock in {\em 2017 IEEE International Conference on Data Mining Workshops
  (ICDMW)}. IEEE, 2017, pp. 1082--1089.

\bibitem{hashlamoun2017mitigation}
W.~Hashlamoun, S.~Brahma, and P.~K. Varshney,
\newblock ``Mitigation of byzantine attacks on distributed detection systems
  using audit bits,''
\newblock {\em IEEE Transactions on Signal and Information Processing over
  Networks}, vol. 4, no. 1, pp. 18--32, 2017.

\bibitem{chen2016optimal}
P.~Chen, Y.~S. Han, H.~Lin, and P.~K. Varshney,
\newblock ``Optimal byzantine attack for distributed inference with m-ary
  quantized data,''
\newblock in {\em IEEE International Symposium on Information Theory (ISIT)}.
  IEEE, 2016, pp. 2474--2478.

\bibitem{rosas2017technological}
F.~Rosas, J.~Hsiao, and K.~Chen,
\newblock ``A technological perspective on information cascades via social
  learning,''
\newblock {\em IEEE Access}, vol. 5, pp. 22605--22633, 2017.

\bibitem{soltanmohammadi2012decentralized}
E.~Soltanmohammadi, M.~Orooji, and M.~Naraghi-Pour,
\newblock ``Decentralized hypothesis testing in wireless sensor networks in the
  presence of misbehaving nodes,''
\newblock {\em IEEE Transactions on Information Forensics and Security}, vol.
  8, no. 1, pp. 205--215, 2012.

\bibitem{kailkhura2014asymptotic}
B.~Kailkhura, Y.~S. Han, .~Brahma, and P.~K. Varshney,
\newblock ``Asymptotic analysis of distributed bayesian detection with
  byzantine data,''
\newblock {\em IEEE Signal Processing Letters}, vol. 22, no. 5, pp. 608--612,
  2014.

\bibitem{zheng2017steady}
X.~Zheng, L.~Xie, and H.~Chen,
\newblock ``Steady-state performance analysis of consensus-based distributed
  detection under sensing data falsification attack,''
\newblock in {\em 9th International Conference on Wireless Communications and
  Signal Processing}. IEEE, 2017, pp. 1--6.

\bibitem{globalsip14}
M.~Doostmohammadian and U.~A. Khan,
\newblock ``On the characterization of distributed observability from first
  principles,''
\newblock in {\em IEEE Global Conference on Signal and Information Processing},
  2014, pp. 914--917.

\bibitem{mo2015physical}
Y.~Mo, S.~Weerakkody, and B.~Sinopoli,
\newblock ``Physical authentication of control systems: Designing watermarked
  control inputs to detect counterfeit sensor outputs,''
\newblock {\em IEEE Control Systems Magazine}, vol. 35, no. 1, pp. 93--109,
  2015.

\bibitem{satchidanandan2016dynamic}
B.~Satchidanandan and P.~R. Kumar,
\newblock ``Dynamic watermarking: Active defense of networked cyber--physical
  systems,''
\newblock {\em Proceedings of the IEEE}, vol. 105, no. 2, pp. 219--240, 2016.

\bibitem{TCNS2020}
M.~Doostmohammadian and N.~Meskin,
\newblock ``Sensor fault detection and isolation via networked estimation:
  Full-rank dynamical systems,''
\newblock {\em IEEE Transactions on Control of Network Systems}, 2020.

\bibitem{chen2018resilient2}
Y.~Chen, S.~Kar, and J.~M.~F. Moura,
\newblock ``Resilient distributed estimation through adversary detection,''
\newblock {\em IEEE Transactions on Signal Processing}, vol. 66, no. 9, pp.
  2455--2469, 2018.

\bibitem{joseph2018observability}
G.~Joseph and C.~R. Murthy,
\newblock ``On the observability of a linear system with a sparse initial
  state,''
\newblock {\em IEEE Signal Processing Letters}, vol. 25, no. 7, pp. 994--998,
  2018.

\bibitem{wakin2010observability}
M.~B. Wakin, B.~M. Sanandaji, and T.~L. Vincent,
\newblock ``On the observability of linear systems from random, compressive
  measurements,''
\newblock in {\em 49th IEEE Conference on Decision and Control}, 2010, pp.
  4447--4454.

\bibitem{li2019distributed}
L.~Li and D.~Li,
\newblock ``A distributed estimation method over network based on compressed
  sensing,''
\newblock {\em International Journal of Distributed Sensor Networks}, vol. 15,
  no. 4, pp. 1550147719841496, 2019.

\bibitem{majidi2017distribution}
M.~Majidi, M.~Etezadi-Amoli, and H.~Livani,
\newblock ``Distribution system state estimation using compressive sensing,''
\newblock {\em International Journal of Electrical Power \& Energy Systems},
  vol. 88, pp. 175--186, 2017.

\bibitem{hamidi2016hybrid}
R.~J. Hamidi, H.~Khodabandelou, H.~Livani, and M.~Sami-Fadali,
\newblock ``Hybrid state estimation using distributed compressive sensing,''
\newblock in {\em IEEE Power and Energy Society General Meeting}, 2016, pp.
  1--5.

\bibitem{xu2015distributed}
S.~Xu, R.~C. De~Lamare, and H.~V. Poor,
\newblock ``Distributed compressed estimation based on compressive sensing,''
\newblock {\em IEEE Signal Processing Letters}, vol. 22, no. 9, pp. 1311--1315,
  2015.

\bibitem{agarwal2020assessing}
P.~Agarwal, M.~Tamer, and H.~Budman,
\newblock ``Assessing observability using supervised autoencoders with
  application to tennessee eastman process,''
\newblock {\em IFAC-PapersOnLine}, vol. 53, no. 2, pp. 206--211, 2020.

\bibitem{wang2020detection}
C.~Wang, S.~Tindemans, K.~Pan, and P.~Palensky,
\newblock ``Detection of false data injection attacks using the autoencoder
  approach,''
\newblock in {\em International Conference on Probabilistic Methods Applied to
  Power Systems (PMAPS)}. IEEE, 2020, pp. 1--6.

\bibitem{wang2018distributed}
J.~Wang, D.~Shi, Y.~Li, J.~Chen, H.~Ding, and X.~Duan,
\newblock ``Distributed framework for detecting pmu data manipulation attacks
  with deep autoencoders,''
\newblock {\em IEEE Transactions on smart grid}, vol. 10, no. 4, pp.
  4401--4410, 2018.

\bibitem{wilson2018deep}
D.~Wilson, Y.~Tang, J.~Yan, and Z.~Lu,
\newblock ``Deep learning-aided cyber-attack detection in power transmission
  systems,''
\newblock in {\em IEEE Power \& Energy Society General Meeting}. IEEE, 2018,
  pp. 1--5.

\bibitem{Khodayar_deep}
M.~Khodayar, G.~Liu, J.~Wang, and M.~E. Khodayar,
\newblock ``Deep learning in power systems research: A review,''
\newblock {\em CSEE Journal of Power and Energy Systems}, vol. 7, no. 2, pp.
  209--220, 2021.

\bibitem{tnse18}
M.~Doostmohammadian, H.~R. Rabiee, H.~Zarrabi, and U.~Khan,
\newblock ``Observational equivalence in system estimation: Contractions in
  complex networks,''
\newblock {\em IEEE Transactions on Network Science and Engineering}, vol. 5,
  no. 3, pp. 212--224, 2018.

\bibitem{icassp13}
M.~Doostmohammadian and U.~A. Khan,
\newblock ``On the distributed estimation of rank-deficient dynamical systems:
  A generic approach,''
\newblock in {\em 38th International Conference on Acoustics, Speech, and
  Signal Processing}, Vancouver, CA, May 2013, pp. 4618--4622.

\bibitem{asilomar11}
M.~Doostmohammadian and U.~A. Khan,
\newblock ``Communication strategies to ensure generic networked observability
  in multi-agent systems,''
\newblock in {\em 45th Annual Asilomar Conference on Signals, Systems, and
  Computers}, Pacific Grove, CA, Nov. 2011, pp. 1865--1868.

\bibitem{khan2014collaborative}
U.~A. Khan and A.~Jadbabaie,
\newblock ``Collaborative scalar-gain estimators for potentially unstable
  social dynamics with limited communication,''
\newblock {\em Automatica}, vol. 50, no. 7, pp. 1909--1914, 2014.

\bibitem{liu-pnas}
Y.~Y. Liu, J.~J. Slotine, and A.~L. Barab\'{a}si,
\newblock ``Observability of complex systems,''
\newblock {\em Proceedings of the National Academy of Sciences}, vol. 110, no.
  7, pp. 2460--2465, 2013.

\bibitem{themis_stochastic}
T.~Charalambous and C.~N. Hadjicostis,
\newblock ``Distributed formation of balanced and bistochastic weighted
  digraphs in multi-agent systems,''
\newblock in {\em European Control Conference}, 2013, pp. 1752--1757.

\bibitem{bay}
J.~Bay,
\newblock {\em Fundamentals of linear state space systems},
\newblock McGraw-Hill, 1999.

\bibitem{usman_cdc:11}
U.~A. Khan and A.~Jadbabaie,
\newblock ``Coordinated networked estimation strategies using structured
  systems theory,''
\newblock in {\em 49th IEEE Conference on Decision and Control}, 2011, pp.
  2112--2117.

\bibitem{kronecker_TSIPN}
M.~Doostmohammadian and U.~A. Khan,
\newblock ``Minimal sufficient conditions for structural
  observability/controllability of composite networks via {K}ronecker
  product,''
\newblock {\em IEEE Transactions on Signal and Information processing over
  Networks}, vol. 6, pp. 78--87, 2020.

\bibitem{rami:97}
L.~El Ghaoui, F.~Oustry, and M.~Ait Rami,
\newblock ``A cone complementarity linearization algorithm for static
  output-feedback and related problems,''
\newblock {\em IEEE Transactions on Automatic Control}, vol. 42, no. 8, pp.
  1171--1176, 1997.

\bibitem{krishnamoorthy2016handbook}
K.~Krishnamoorthy,
\newblock {\em Handbook of statistical distributions with applications},
\newblock CRC Press, 2016.

\bibitem{khan2013secure}
U.~Khan and A.~Stankovic,
\newblock ``Secure distributed estimation in cyber-physical systems,''
\newblock in {\em IEEE International Conference on Acoustics, Speech and Signal
  Processing (ICASSP)}, 2013, pp. 5209--5213.

\bibitem{TNSE19}
M.~Doostmohammadian and U.~A. Khan,
\newblock ``On the complexity of minimum-cost networked estimation of
  self-damped dynamical systems,''
\newblock {\em IEEE Transactions on Network Science and Engineering}, vol. 7,
  no. 3, pp. 1891--1900, 2019.

\bibitem{spl18}
M.~Doostmohammadian, H.~R. Rabiee, and U.~A. Khan,
\newblock ``Structural cost-optimal design of sensor networks for distributed
  estimation,''
\newblock {\em IEEE Signal Proc. Letters}, vol. 25, no. 6, pp. 793--797, 2018.

\bibitem{murota}
K.~Murota,
\newblock {\em Matrices and matroids for systems analysis},
\newblock Springer, 2000.

\bibitem{icas21_contraction}
M.~Doostmohammadian, T.~Charalambous, M.~Shafie-khah, H.~R. Rabiee, and U.~A.
  Khan,
\newblock ``Analysis of contractions in system graphs: Application to state
  estimation,''
\newblock in {\em 1st IEEE International Conference on Autonomous Systems
  (ICAS)}, 2021.

\end{thebibliography}

\begin{IEEEbiography}[{\includegraphics[width=1.1in]{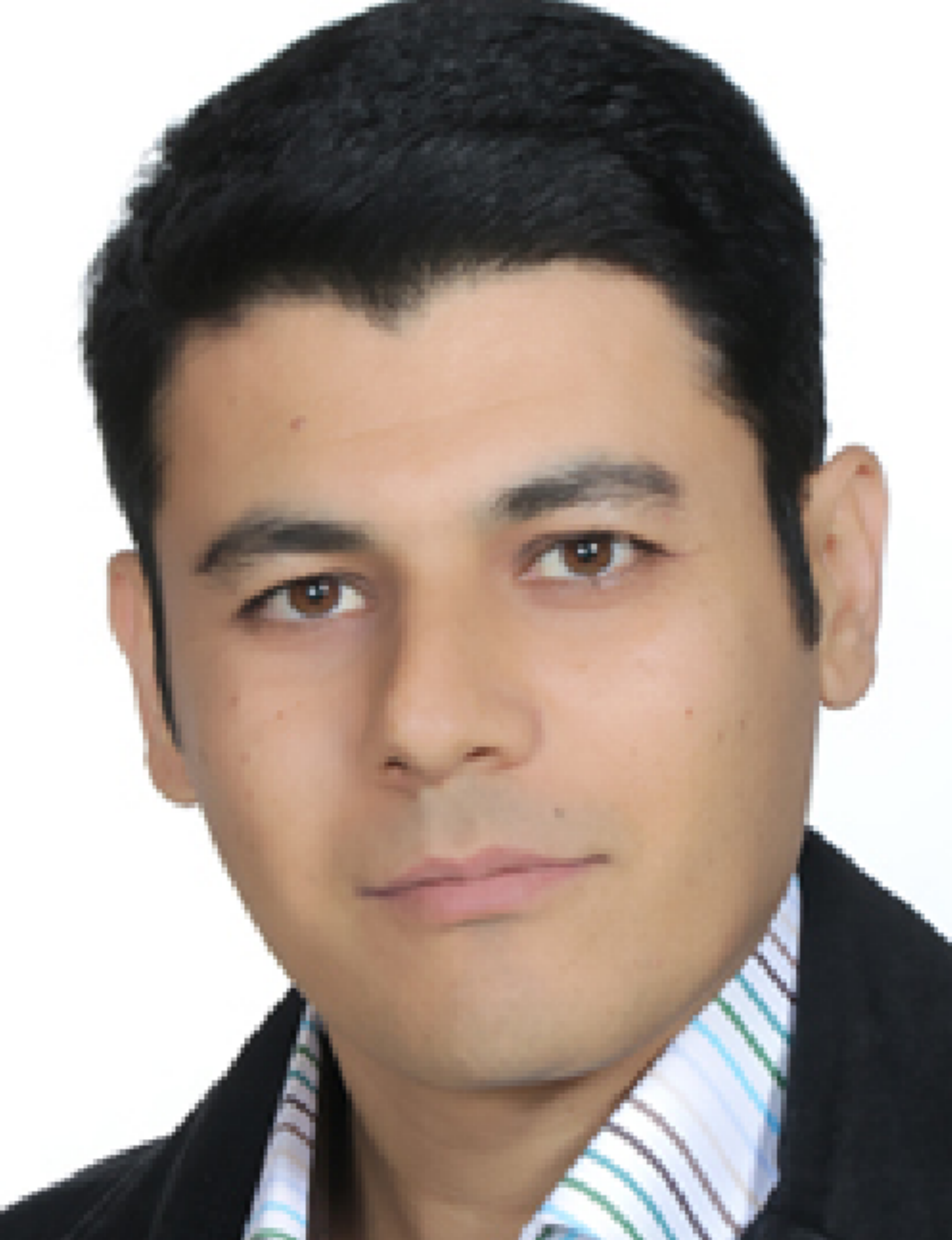}}]{Mohammadreza~Doostmohammadian}
	received his B.Sc. and M.Sc. in Mechanical Engineering from Sharif University of Technology (SUT), and Ph.D. in Electrical Engineering from Tufts University. He was a postdoc at AICT, School of Computer Engineering, SUT and a researcher at ITRC. Recognition of his work includes IEEE JSTSP journal cover and IEEE MSC09 and ICNSC14 conference awards.
	Currently, he is an Assistant Professor of Mechatronics at Semnan University and a   researcher with  Aalto University. His general research interest includes distributed optimization, control, and estimation over networks. 
	He was the chair of the robotics and control session at ISME-2018 conference.
\end{IEEEbiography}
\vspace{-0.6cm}
\begin{IEEEbiography}[{\includegraphics[width=1in]{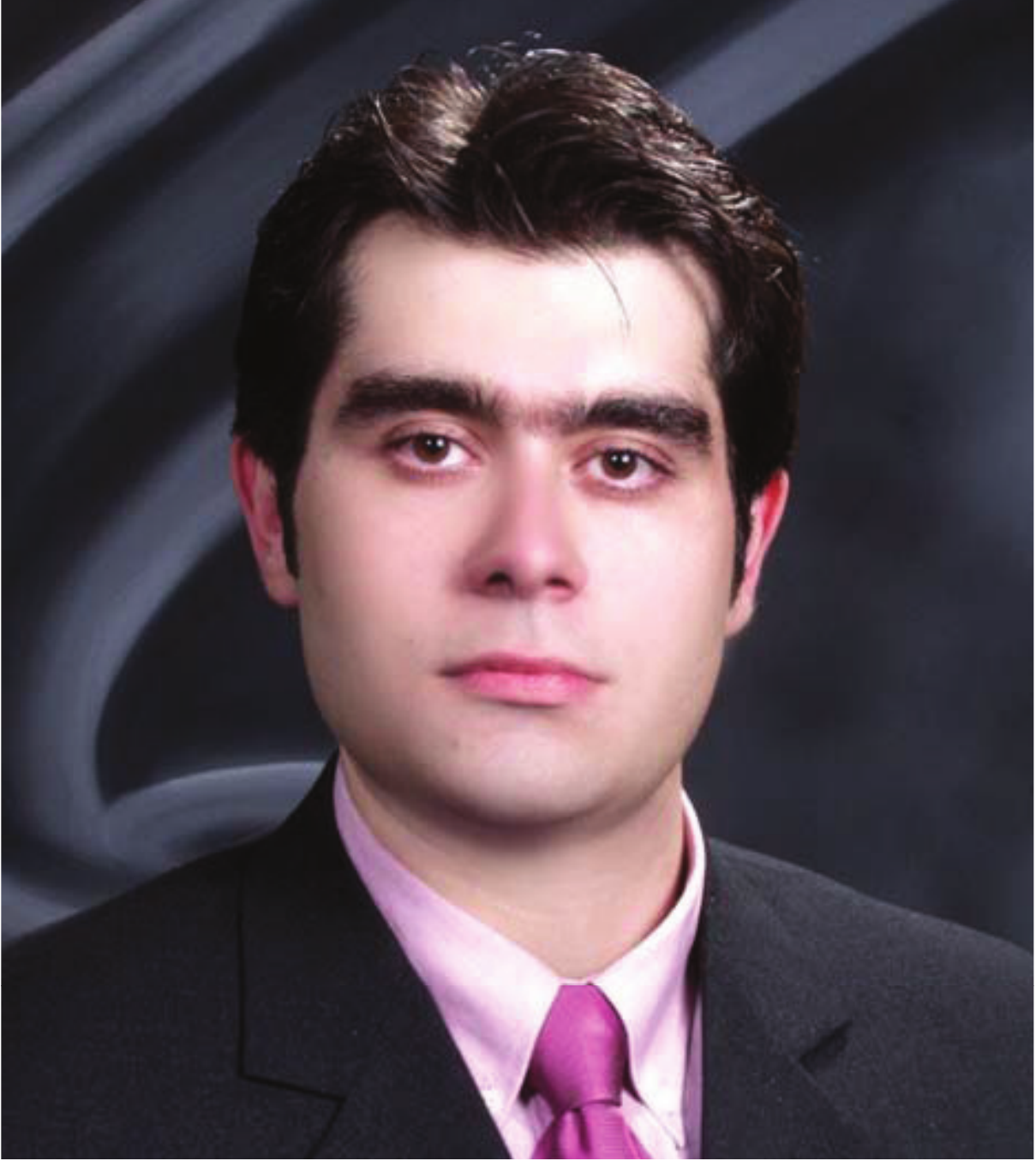}}]{Houman Zarrabi}
	received his Ph.D. from Concordia University in Montreal, Canada in 2011. Since then he has been involved in various industrial and research projects. His main expertise includes IoT, M2M, CPS, big data, embedded systems, and VLSI. He is currently the national IoT program director and assistant professor at Iran Telecommunication Research Center (ITRC).
\end{IEEEbiography}
\vspace{-0.6cm}
\begin{IEEEbiography}[{\includegraphics[width=1.1in]{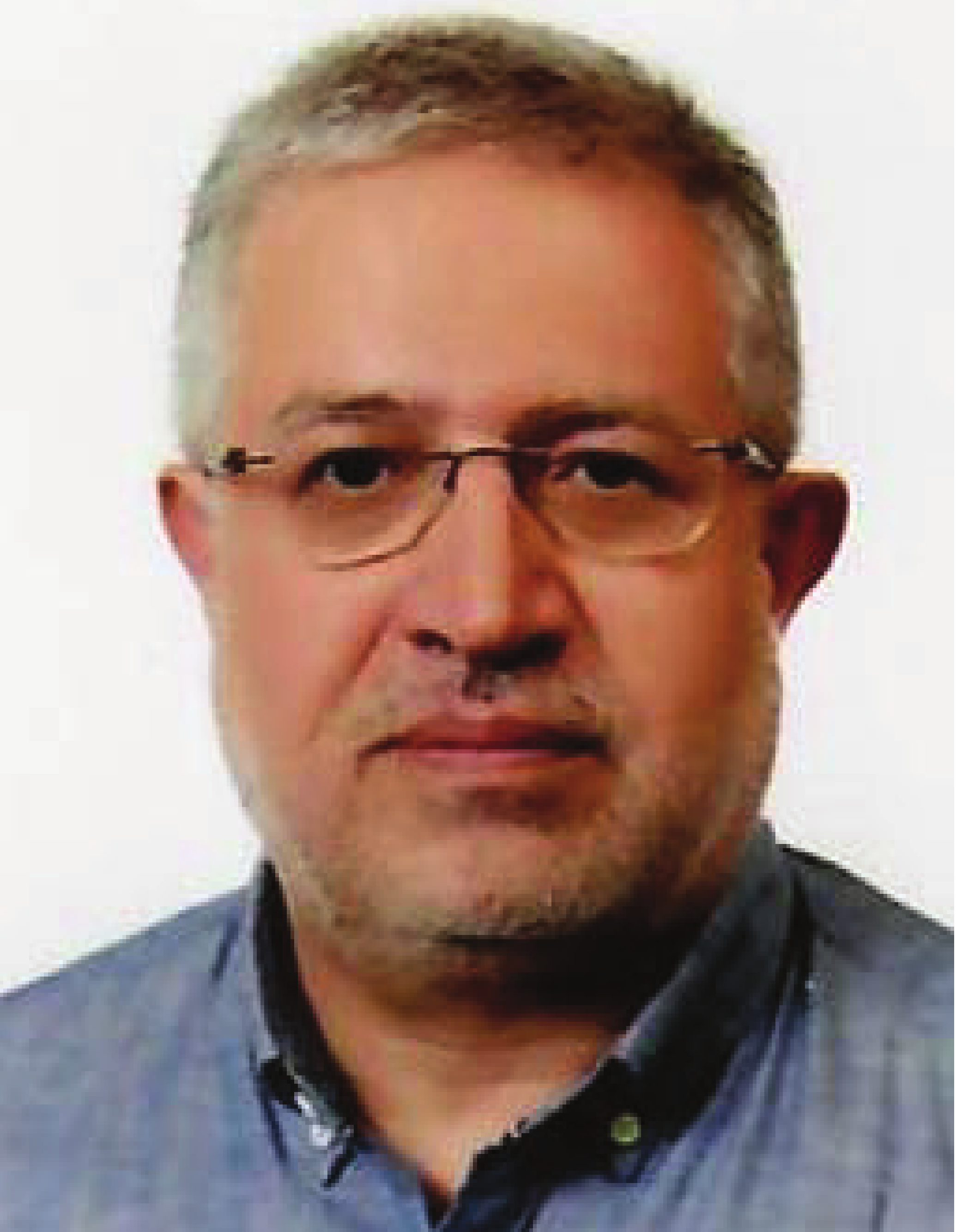}}]{Hamid R. Rabiee}
	received his BS and MS degrees in Electrical Engineering from CSULB, his EEE degree and his Ph.D. in Electrical and Computer Engineering from USC (1993), and from Purdue University (1996). He was with AT\&T Bell Laboratories, Intel Corporation as a Senior Software Engineer, and with PSU, OGI, and OSU as an adjunct professor. He was also a visiting professor at the Imperial College of London for the 2017-2018 academic year. He is the founder of AICT, SATI, DML, VASL, BCB, and Cognitive Neuroengineering Research Center. He is currently a Professor of Computer Engineering at SUT. 
\end{IEEEbiography}
\vspace{-0.6cm}
\begin{IEEEbiography}[{\includegraphics[width=1.1in]{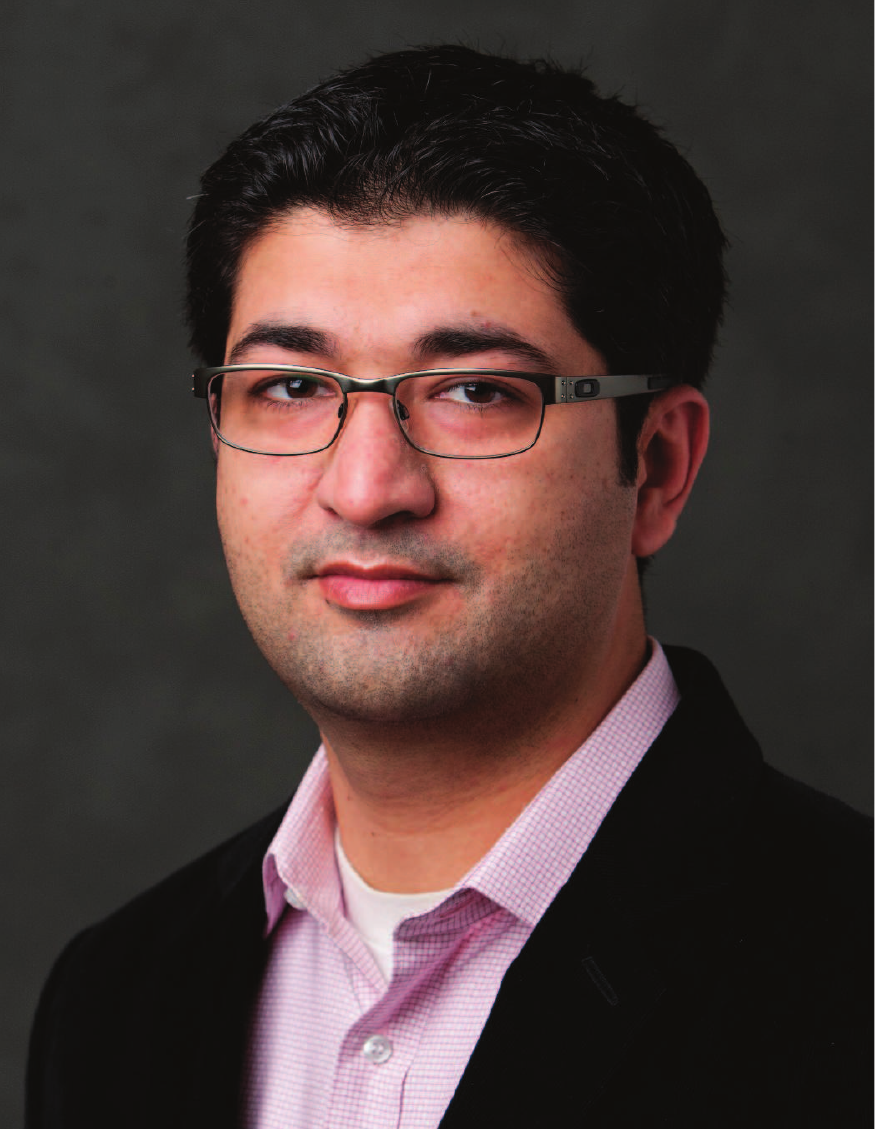}}]{Usman A. Khan} (Senior Member, IEEE) received the B.S. degree from the University of Engineering and Technology, Lahore, the M.S. degree from the
University of Wisconsin–Madison, and the Ph.D. degree from
Carnegie Mellon University, all in electrical and computer
engineering. He held a postdoc position at the GRASP Laboratory, UPenn.  He was a Visiting Professor with KTH, and currently is an
Associate Professor of electrical and computer engineering (ECE) with Tufts University,  where he is also an Adjunct
Professor of computer science. 
Recognition of his work includes the prestigious NSF Career Award, several NSF REU awards, an IEEE journal cover, three IEEE best student paper awards.

\end{IEEEbiography}
\vspace{-0.6cm}
\begin{IEEEbiography}[{\includegraphics[width=1.1in]{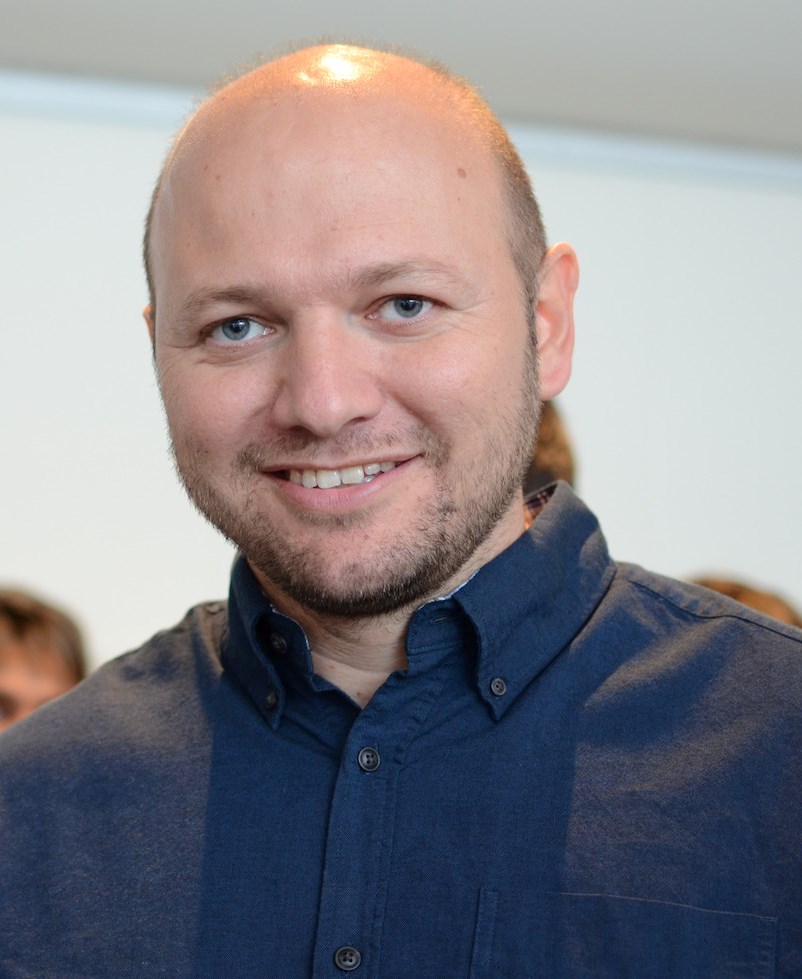}}]{Themistoklis Charalambous}
received his BA and M.Eng in Electrical and Information Sciences from Trinity College, Cambridge University. He completed his Ph.D. studies in the Control Laboratory, of the Engineering Department, Cambridge University. He joined the Human Robotics Group as a Research Associate at Imperial College London for an academic year and worked as a Visiting Lecturer at the Department of Electrical and Computer Engineering, University of Cyprus. He worked as a postdoc at the 
Department of Automatic Control of the School of Electrical Engineering at KTH and Department of Electrical Engineering at Chalmers University of Technology. 
Since September 2018 he was nominated Research Fellow of the Academy of Finland and since July 2020 he is a tenured Associate Professor of Electrical Engineering at Aalto University.
\end{IEEEbiography}

\end{document}